\documentclass[12pt,onecolumn,a4paper]{article}
\usepackage[breaklinks, colorlinks = true, linkcolor = darkblue, urlcolor  = darkblue, citecolor = red, anchorcolor = darkblue]{hyperref}
\usepackage[left=1in,right=1in,top=1in,bottom=1in]{geometry}
\usepackage{cite}
\usepackage{doi}

\hypersetup{breaklinks, colorlinks = true, linkcolor = blue, urlcolor  = blue, citecolor = ForestGreen, anchorcolor = blue}
\usepackage[noblocks]{authblk}
\usepackage[dvipsnames]{xcolor}
\usepackage{comment}
\usepackage{bm}

\usepackage{subcaption}
\usepackage{enumerate}

\usepackage{amsfonts}
\usepackage{amsmath}
\usepackage{graphicx}
 
\usepackage{mathtools}
\mathtoolsset{showonlyrefs}

\usepackage{dcolumn}

\usepackage{enumerate}
\usepackage{amssymb}
\usepackage{amsthm}
\usepackage{calc} 
\usepackage{accents}
\usepackage{dsfont}

\usepackage{multirow}

\newcommand*{\cA}{\mathcal{A}} 
\newcommand*{\cB}{\mathcal{B}}
\newcommand*{\cC}{\mathcal{C}}

\newcommand*{\cG}{\mathcal{G}}
\newcommand*{\cH}{\mathcal{H}}
\newcommand*{\cI}{\mathcal{I}}

\newcommand*{\cK}{\mathcal{K}}
\newcommand*{\cL}{\mathcal{L}}
\newcommand*{\cM}{\mathcal{M}}
\newcommand*{\cN}{\mathcal{N}}
\newcommand*{\cD}{\mathcal{D}}

\newcommand*{\cR}{\mathcal{R}}

\newcommand*{\cP}{\mathcal{P}}

\newcommand*{\cU}{\mathcal{U}}
\newcommand*{\cT}{\mathcal{T}}
\newcommand*{\cV}{\mathcal{V}}
\newcommand*{\cW}{\mathcal{W}}

\newcommand*{\tr}{\mathop{\mathrm{tr}}\nolimits}

\newcommand{\id}{\operatorname{id}}

\newcommand{\1}{\mathds{1}}

\newcommand{\und}[1]{\underline{#1}}

\newcommand{\erfc}{\operatorname{erfc}}

\makeatletter
\newtheorem*{rep@theorem}{\rep@title}
\newcommand{\newreptheorem}[2]{%
\newenvironment{rep#1}[1]{%
 \def\rep@title{#2 \ref{##1}}%
 \begin{rep@theorem}}%
 {\end{rep@theorem}}}
\makeatother

\newtheorem{theorem}{Theorem} 
\newreptheorem{theorem}{Theorem}
\newtheorem{lemma}[theorem]{Lemma}
\newreptheorem{lemma}{Lemma}

\newreptheorem{claim}{Claim}
\newtheorem{definition}[theorem]{Definition}
\newreptheorem{definition}{Definition}
\newtheorem{remark}[theorem]{Remark}
\newreptheorem{remark}{Remark}

\newreptheorem{corollary}{Corollary}

\newreptheorem{observation}{Observation}
\newtheorem{proposition}[theorem]{Proposition}
\newreptheorem{proposition}{Proposition}

\newcommand{\myacknowledgements}{\begin{center}{\bf Acknowledgements}\end{center}\par}

\usepackage{amsfonts}

\newcommand{\ket}[1]{|#1\rangle}
\newcommand*{\mylabel}[1]{\label{#1}}
\newcommand{\bra}[1]{\langle#1|}

\newcommand{\proj}[1]{|#1\rangle\!\langle#1|}

\newcommand{\assign}{:=}

\usepackage[symbol]{footmisc}

\makeatletter
\renewcommand*{\@fnsymbol}[1]{\ensuremath{\ifcase#1\or \circ \or \diamond\or \natural\or
    \sharp\or \dagger\or \ddagger\or \|\or **\or \dagger\dagger
    \or \ddagger\ddagger \else\@ctrerr\fi}}
\makeatother

\newcommand*{\Ck}{\mathbb{C}^k}

\begin{document}

\title{\LARGE Infinite-Dimensional Programmable \protect\\ Quantum Processors}

\author[1,2]{Martina Gschwendtner\thanks{martina.gschwendtner@tum.de}}
\author[3,4]{Andreas Winter\thanks{andreas.winter@uab.cat}}
\renewcommand\Affilfont{\fontsize{10}{11.8}}
\affil[1]{\footnotesize Munich Center for Quantum Science and Technology (MCQST), 80799 M\"unchen, Germany}
\affil[2]{\footnotesize Zentrum Mathematik, Technical University of Munich, 85748 Garching, Germany}
\affil[3]{\footnotesize Instituci\'o Catalana de Recerca i Estudis Avan\c{c}ats (ICREA),\protect\\ Pg.~Lluis Companys, 23, 08001 Barcelona, Spain}
\affil[4]{\footnotesize Grup d'Informaci\'{o} Qu\`{a}ntica, Departament de F\'{\i}sica,\protect\\ Universitat Aut\`{o}noma de Barcelona, 08193 Bellaterra (Barcelona), Spain}

\date{\small 13 July 2021}

\maketitle

\begin{abstract}
A universal programmable quantum processor uses ``program'' quantum states to apply an arbitrary quantum channel to an input state. We generalize the concept of a finite-dimensional programmable quantum processor to infinite dimension assuming an energy constraint on the input and output of the target quantum channels. By proving reductions to and from finite-dimensional processors, we obtain upper and lower bounds on the program dimension required to approximately implement energy-limited quantum channels. 
In particular, we consider the implementation of Gaussian channels. Due to their practical relevance, we investigate the resource requirements for gauge-covariant Gaussian channels. Additionally, we give upper and lower bounds on the program dimension of a processor implementing all Gaussian unitary channels. These lower bounds rely on a direct information-theoretic argument, based on the generalization from finite to infinite dimension of a certain ``replication lemma'' for unitaries.
\end{abstract}

\section{Introduction}

A programmable quantum processor takes an input state and applies a quantum channel to it, controlled by a program state that contains all relevant information for the implementation. This concept, introduced by Nielsen and Chuang~\cite{Nielsen97}, is inspired by the von Neumann architecture of classical computers and universal Turing machines, which postulate a single device operating on data using a ``program'', which ultimately is just another kind of data. 
The main result of Ref.~\cite{Nielsen97} is the \emph{No-Programming Theorem}, stating that a universal programmable quantum processor, i.e., one capable of the implementation of any quantum channel exactly with finite program dimension, does not exist. If one relaxes the requirement to only approximate implementation, a trade-off between the accuracy of the implementation and the size of the program register occurs.

Exact and approximate quantum processors have been extensively studied, with respect to the required resources, which here are identified with the dimension of the program register. Many efforts have been made to obtain optimal scaling in different regimes. Recently, Yang \emph{et al.}~\cite{Renner20} closed the gap for the universal implementation of unitaries, providing essentially the optimal scaling of the program dimension and the accuracy for a given finite-dimensional Hilbert space on which the unitaries act.

Since infinite-dimensional (also known as continuous-variable) systems are fundamental in quantum theory, and are gaining more and more attention in quantum communication~\cite{Weedbrook12,GGCH}, quantum cryptography~\cite{Grosshans-et-al,Leverrier}, and quantum computing~\cite{KLM,BosonSampl,GKP}, here we investigate programmable quantum processors for continuous-variable systems.

\medskip
\emph{Results.} In the present paper, we define programmable quantum processors with infinite-dimensional input and output, assuming an energy constraint. This means that we seek to approximately implement arbitrary energy-limited unitary quantum channels (which are those that map energy-bounded states to energy-bounded states) with finite program dimension, the accuracy of the implementation judged using the energy-constrained diamond norm~\cite{Shirokov18,Winter17} (cf.~Refs.~\cite{Lupo17,Pirandola17} for a closely related definition) as distance measure. After Definition \ref{epsilonEPQP}, we argue the intrinsic necessity of both the use of the energy-constrained diamond norm and energy limitation of the channels. 

Our first group of main results concerns the resource requirements, determining upper and lower bounds on the program dimension. We achieve this by relating the performance of infinite-dimensional approximate programmable quantum processors to that of their finite-dimensional counterparts in Theorems~\ref{maintheorem} and~\ref{thm:constructionforlowerbounds}. In Theorem~\ref{maintheorem}, we construct an infinite-dimensional programmable quantum processor assuming an existing finite-dimensional one whereas Theorem~\ref{thm:constructionforlowerbounds} states that a finite-dimensional programmable quantum processor can be constructed assuming an existing infinite-dimensional one. We establish thus a fundamental link between discrete and continuous quantum systems in the context of programmable quantum processors. These results allow us to import known upper and lower bounds on the program register from finite-dimensional processors. The upper bounds we obtain are summarized in Table~\ref{tab:finiteupperbounds} whereas the lower bounds can be found in Table~\ref{tab:finitelowerbounds}.

The case of unitary channels, in the exact case covered by Nielsen and Chuang's No-Programming Theorem, was analyzed by Yang \emph{et al.}~\cite{Renner20} for the approximate implementation and using an information-theoretic approach. In Lemma \ref{lemma:E-Yang} we generalize their central tool, the recycling of the program register to implement the same unitary multiple times, to infinite-dimensional unitaries and energy-constrained diamond norm approximation. To do so, we also introduce a multiply energy-constrained diamond norm with several simultaneous constraints rather than a single one, in Eq.~\eqref{eq:multiply-E-diamond}. At the end of Section \ref{sec:recycling}, we outline what Yang \emph{et al.}'s information-theoretic lower bound strategy looks like in infinite dimension. The approach of lower bounding the program dimension by relating the Holevo information of program and (compressed) output ensembles is then used in the subsequent consideration of various classes of Gaussian Bosonic channels.

Looking beyond fully universal processors, we furthermore study the approximate implementation of energy-limited gauge-covariant Gaussian channels of a single Bosonic mode. Gauge-covariant channels play an important role physically, since they describe the consequences of attenuation of signals and addition of noise in communication schemes. Furthermore, they preserve the class of thermal Gaussian states~\cite{Weedbrook12}.
Using $\epsilon$-net constructions, we provide upper bounds on the program dimension in Theorem~\ref{thm:upperboundsgaugecovariant}, which diverge with the accuracy of implementation. Due to the energy limitation, our lower bounds for gauge-covariant channels rely only on lower bounds for the implementation of the attenuator channels and of the phase unitaries. The concrete bounds are stated in Theorems~\ref{thm:Gaussianlowerbounds} for the phase unitaries and Theorem \ref{thm:attenuator-lowerbounds} for attenuators, the former diverging with the implementation accuracy.

Another case study is the universal implementation of energy-limited Gaussian unitaries of any finite number of Bosonic modes. We show that there exists an infinite-dimensional programmable quantum processor whose program register can be upper bounded in terms of the approximation error $\epsilon$ in Theorem~\ref{thm:upperboundsGaussianunitary}. The upper bound we obtain, using an $\epsilon$-net construction, diverges as inverse polynomial in $\epsilon$.
Due to the phase unitaries treated earlier, we show that the program register of every inifinte-dimensional programmable quantum processor can be lower bounded by a term diverging with the inverse of $\epsilon$ in Theorem~\ref{thm:Gaussianlowerbounds-M-modes} with a different power, though.

\medskip
\emph{Context.} Nonuniversal programmable quantum processors, i.e., those implementing (either exactly or approximately) only a prescribed set of, rather than all, quantum channels between two given systems have occurred in various other contexts different from quantum computing, where channel simulation can help reduce the analysis of quantum channels to that of the corresponding program state(s). For example, the quantum and private capacities of a quantum channel, assisted by local operations and classical communication (LOCC), are the same as those of its Choi state, if the channel can be implemented by a processor built from LOCC elements with the Choi state as the program~\cite{Bennett96, Pirandola17}. Similarly, if a family of channels can be implemented by the same processor using their Choi states as programs, adaptive strategies to distinguish or estimate the channels cannot be better than nonadaptive ones, for any finite number of channel uses and asymptotically, and in fact boil down to the much better-understood discrimination of the Choi states~\cite{Lupo17,Takeoka16}, see also Refs.~\cite{Pirandola-et-al-2018,Berta18}. 

\medskip
\emph{Structure.}
In Section \ref{sec:prelim} we start by introducing basic notions for infinite-dimensional quantum information; 
in Section \ref{section2} we prove two reductions, first of an infinite-dimensional universal programmable quantum processor for energy-limited channels (unitaries) to a finite-dimensional one for arbitrary channels (unitaries), and second vice versa, and can thus import known program dimension bounds (upper and lower) from the finite-dimensional world; 
in Section \ref{sec:recycling} we prove
that a processor implementing unitaries approximately can implement them repeatedly by reusing the same program successively, and sketch how this leads to information-theoretic lower bounds on program registers (cf. Yang \emph{et al.}~\cite{Renner20});
in Section \ref{sec:gaussian} we turn our attention to quantum processors for gauge-covariant one-mode Gaussian channels, on the one hand, and multimode Gaussian unitaries on the other, proving upper and lower bounds on the program dimension in both settings; 
in Section \ref{sec:conclusion} we conclude, discussing future directions and open questions.

\section{Notation and preliminaries}
\label{sec:prelim}
Let $\cH$ be a separable Hilbert space and denote $\cT (\cH)$ the Banach space of all trace-class operators acting on $\cH$, equipped with the trace norm $\| A \|_1 = \tr \sqrt{A^\ast A}$, 
\begin{equation}
  \cT (\cH) = \{ A: \cH \to \cH \text{ such that } \|A \|_1 < \infty \}.
\end{equation}
The set of quantum states, which we identify with density operators (positive semidefinite operators in $\cT (\cH)$ with unit trace) is denoted $\cD(\cH)$, and the subset of pure quantum states by $\cD_P(\cH)$. The set of all quantum channels is denoted $\text{CPTP}(\cH_1,\cH_2)$.
A quantum channel is a completely positive (CP) and trace-preserving (TP) linear map (superoperator) $\Phi : \cT(\cH_1) \to \cT(\cH_2)$ for separable Hilbert spaces $\cH_1$ and $\cH_2$. 
It can equivalently be described by its adjoint map on the dual spaces of bounded operators, $\Phi^\ast:\cB(\cH_2)\to \cB(\cH_1)$, which is a completely positive and unit preserving (CPUP) map characterized uniquely by the duality relation
\[
  \tr \Phi(\alpha)B = \tr\alpha\Phi^\ast(B) \quad\forall\alpha\in \cT(\cH_1),\, B\in \cB(\cH_2).
\]

Let $H$ be a positive semidefinite Hamiltonian $H\geq 0$ with discrete spectrum. We denote its spectral decomposition as $H=\sum_{n=0}^\infty e_n P_n$, in the sense of the spectral theorem for unbounded self-adjoint operators, see Refs.~\cite{vonNeumann:book, Hall:QM-x-Mathematicians}: it amounts to  convergence of $\sum_n e_n \bra{\phi}P_n\ket{\phi'}$ to $\bra{\phi}H\ket{\phi'}$ for all $\ket{\phi},\ket{\phi'}$ in the domain of $H$. We additionally assume finite degeneracy of all eigenvalues, i.e., $\dim P_n<\infty$ for all $n$ (we call such Hamiltonians \emph{finitary}). 
In the following, we always assume that $0$ is the smallest eigenvalue of $H$ (such Hamiltonians are called \emph{grounded}). 

Throughout this paper, we consider energy constraints on the input and output, of the form $\tr \rho H_1 \leq E_1$ and $\tr \Phi(\rho) H_2 \leq E_2$, respectively, where $H_1\geq 0$ is a finitary and grounded Hamiltonian as above (i.e., with discrete spectrum, smallest eigenvalue $0$ and finite degeneracy of every eigenvalue) acting on $\cH_1$, and $H_2 \geq 0$ likewise on $\cH_2$. 
Thus, let us define energy-limited quantum channels, which map energy-bounded states to energy-bounded states.

\begin{definition}[{$(\alpha, \beta)$-energy-limited quantum channel~\cite{Winter17}}]
\mylabel{definitionenergylimitedchannel}
Given two positive semidefinite Hamiltonians $H_1\geq 0$ and $H_2\geq 0$ on $\cH_1$ and $\cH_2$, respectively, a quantum channel $\Phi:\cT(\cH_1)\to \cT(\cH_2)$ is called \emph{($\alpha, \beta $)-energy-limited} if for all
$\rho\in\cD(\cH_1)$ with $\tr\rho H_1 < \infty$, it holds $\tr\Phi(\rho)H_2 < \infty$, and in fact
\begin{equation}
  \tr \Phi(\rho) H_2 \leq \alpha \tr \rho H_1 + \beta.
\end{equation}
This can be expressed equivalently as
$\Phi^\ast(H_2) \leq \alpha H_1 + \beta\mathds{1}$, 
using the adjoint cptp map $\Phi^\ast:\cB(\cH_2)\to \cB(\cH_1)$, which however has to be interpreted suitably, given that the adjoint map is a priori only defined on bounded operators and the same for the semidefinite operator order -- see the following explanation.
\end{definition}

Since we have the spectral decomposition of $H_2 = \int_0^\infty e\,{\rm d}P(e)$, with the spectral measure ${\rm d}P$ on $[0;\infty)$, we can define a positive operator-valued measure (POVM) ${\rm d}M$ on $[0;\infty)$ by letting $M(I) = \Phi^*(P(I))$ for any measurable set $I \subset [0;\infty)$, and let $\Phi^*(H_2) = \int_0^\infty e\,{\rm d}M(e)$, assuming that there is a dense domain on which the convergence with respect to test vectors holds. Equivalently, we can consider the bounded operators $H_2(E) := \int_0^E e\,{\rm d}P(e) + E P\bigl((E;\infty)\bigr)$, which form an increasing family that converges to $H_2$ -- as usual, in the weak operator sense with respect to pairs of test vectors from the domain of $H_2$; then, $\Phi^*(H_2(E))$ is an increasing family of bounded positive semidefinite operators on $\cH_1$, whose limit, if it exists, is called $\Phi^*(H_2)$. Note that in either way, the definition makes sure that the result, if densely defined, is an essentially self-adjoint operator.

As for the operator order, $A\leq B$ for bounded operators means that $B-A$ is positive semidefinite. For unbounded $A$ and $B$, we only define it if both operators are self-adjoint, in particular with dense domains $D(A)$ and $D(B)$, respectively. Then, $A\leq B$ is defined as meaning $D(A) \supset D(B)$ and $\bra{\psi}A\ket{\psi} \leq \bra{\psi}B\ket{\psi}$ for all $\ket{\psi}\in D(B)$. Notice that we employ this notation for a positive semidefinite $B$, where $\ket{\psi}\not\in D(B)$ amounts to $\bra{\psi}B\ket{\psi}=+\infty$, thus automatically fulfilling the inequality also in that case.

\medskip
So as to allow trivial channels, such as the identity and the constant channel mapping every input state to the ground state of the output space, we always assume $\alpha\geq 1$ and $\beta\geq 0$. 

The diamond norm is a well-motivated metric to differentiate quantum channels in terms of their statistical distinguishability. Since we consider energy constraints on the channel input, we are motivated to use the energy-constrained diamond norm instead.

\begin{definition}[{Energy-constrained diamond norm~\cite{Winter17,Shirokov18}}]
\label{energydiamondnorm}
Let $H_1\geq 0$ be a grounded Hamiltonian on $\cH_1$, and $E>0$. For a Hermitian-preserving map $\Phi : \cT (\cH_1) \to \cT (\cH_2)$, define the energy-constrained diamond norm (more precisely: $E$-constrained diamond norm) as
\begin{equation}
\| \Phi \|_\diamond^{E} = \sup_{\cH_R} \sup_{\substack{\rho \in \cD (\cH_1\otimes \cH_R)\\  \tr \rho ( H_1 \otimes \mathds{1}_R) \leq E}} \| (\Phi \otimes \operatorname{id}_R)\rho\|_1.
\end{equation}
\end{definition}

From the definition, one may without loss of generality restrict the test states $\rho\in\cD(\cH_1\otimes\cH_R)$ to be pure states, and $\cH_R$ may be assumed isomorphic to $\cH_1$.
For the trivial Hamiltonian $H_1=0$ this definition reduces to the usual diamond norm. On the other hand, Shirokov~\cite{Shirokov18} showed for our class of finitary and grounded Hamiltonians, that the energy-constrained diamond norms $\|\cdot\|_\diamond^E$ all induce the strong topology. 
For more details and properties of the energy-constrained diamond norm we refer to Refs.~\cite{Winter17,Shirokov18}. 
Note that a slightly different notion of energy-constrained diamond norm was considered earlier in Refs.~\cite{Lupo17} and~\cite{Pirandola17}. 

In the following, we adapt the definition of an \emph{$\epsilon$-universal programmable quantum processor ($\epsilon$-UPQP)}, originally discussed for finite-dimensional systems, to infinite dimension. 

\begin{definition}[$\epsilon\text{-PQP}_{\cC}$, cf. Ref.~\cite{Kubicki19}]
\label{epsilon-PQP}
Let $\cH_1$ and $\cH_2$ be separable Hilbert spaces. Then, we call $\cP \in \text{CPTP}(\cH_1 \otimes \cH_P,\cH_2)$, with finite-dimensional $\cH_P$, an \emph{$\epsilon$-programmable quantum processor for a set $\cC \subset \text{CPTP}(\cH_1,\cH_2)$} of channels ($\epsilon\text{-PQP}_{\cC}$), if for every cptp map $\Phi \in \cC$ there exists a state $\pi_\Phi \in \cD(\cH_P)$ such that
\begin{equation}
\frac12 \| \cP(\cdot \otimes \pi_\Phi)
           - \Phi(\cdot) \|_\diamond \leq \epsilon.
\end{equation}
To address the Hilbert spaces $\cH_1$, $\cH_2$ and $\cH_P$, we refer to the former two as the \emph{input} and \emph{output registers}, to the latter as \emph{program register}. 
We say that the processor $\cP$ $\epsilon$-\emph{implements} the class $\cC$ of channels, leaving out the reference to $\epsilon$ when it is $0$.
\end{definition}

When $\cC = \text{CPTP}(\cH_1,\cH_2)$, we call the processor \emph{universal} and denote it as $\epsilon\text{-PQP}_\text{CPTP}$. 
Note that allowing mixed states in the program register is essential, since it allows for example the programming of all depolarizing channels using a qubit $\cH_P$. 
On the other hand, we can always replace a mixed program state $\pi_\Phi$ by a suitable purification on  $\cH_P\otimes\cH_P$ to obtain a pure program state at the expense of increasing (squaring) the program dimension, and accordingly modify the processor to one acting on $\cH_1\otimes\cH_P^{\otimes 2}$.
Another important special case, that has been considered before, is that $\cH_1=\cH_2=\cH$ is a finite-dimensional Hilbert space, and $\cC = \cU(\cH)$ the set of all unitaries, or rather the channels defined by conjugation with unitaries, which has been addressed as ``universal'' in the literature, but which -- in view of the restriction to unitaries -- we want to call \emph{unitary-universal} and denote $\epsilon\text{-PQP}_\cU$; 
in the literature this class is referred to as $\epsilon$-UPQP, a ``universal'' programmable quantum processor. We consider a $d_1$-dimensional input space, a $d_2$-dimensional output space and an $d_P^\infty$-dimensional program register. In the unitary case, the program states are customarily assumed to be pure, in accordance with the previous literature. 

\medskip
For the exact case, i.e., $\epsilon$=0, Nielsen and Chuang proved the No-Programming Theorem~\cite{Nielsen97}. They consider a processor which implements the set $\cC$ of the channels generated by conjugation with unitaries $U_1, \ldots, U_N$ perfectly.

\begin{theorem}[No-Programming~\cite{Nielsen97}] \label{NoPro}
Let $U_1, \ldots , U_N \in \cU(\mathbb{C}^d)$ be distinct unitary operators (up to a global phase) which are implemented by some programmable quantum processor. Then,
\begin{enumerate}
\item the corresponding programs $\ket{P_1}, \ldots, \ket{P_N}$, which are states of the program register $\cH_P$, are mutually orthogonal; 
\item the program register is at least $N$-dimensional, or in other words, it contains at least $\log_2 N$ qubits.
\end{enumerate}
\end{theorem}

\noindent While the theorem was originally proved for finite-dimensional $\cH$, it is not difficult to see that its proof extends to separable Hilbert spaces. 

This result shows that no exact universal quantum processor with finite-dimensional program register exists. Indeed, every unitary operation which is implemented by the processor requires an extra dimension of the program Hilbert space. Since there are infinitely, in fact uncountably, many unitary operations, no universal quantum processor with finite-dimensional (or indeed separable Hilbert) space exists. Note that with a separable Hilbert space, one can however approximate every unitary channel arbitrarily well, by choosing $\cC$ as a countable, dense set of unitaries (see the following construction). On the other hand, using a $d$-dimensional program register, up to $d$ unitary operations, which are distinct up to a global phase, can be implemented by a series of controlled unitary operations.
Therefore, we are interested in $\epsilon>0$ and in particular how the program register depends on the accuracy $\epsilon$.

Note that if $\cC = \{\Phi_i \}_{i=1}^K$ is a finite set of channels, then we can construct the processor that implements those channels exactly with memory dimension equal to the cardinality $K$ of the set in the following way. We encode the specifying index $i$ from the set of channels into the program state $\pi_\Phi$. The following processor implements the channels  $\Phi_{i}$ from the set $\cC$ exactly: 
\begin{equation}
\cP (\rho \otimes \pi_\Phi)
  \assign \sum_{i} \Phi_{i} (\rho) \bra {i}\pi_\Phi \ket {i},
\end{equation}
which clearly satisfies
\begin{equation}
\cP (\rho \otimes \proj{i}) = \Phi_{i} (\rho).
\end{equation}
The program dimension $d_P^\infty$ is equal to the cardinality $K$ of the set $\cC$, thus meeting the lower bound from the Nielsen/Chuang No-Programming Theorem \ref{NoPro} (for unitary channels). Since we use this construction several times throughout the paper, we refer to it as the \emph{processor-encoding technique (PET)}.

In the present paper we are interested in programmable quantum processors that $\epsilon$-implement $(\alpha,\beta)$-energy-limited quantum channels between infinite-dimensional systems $\cH_1$ and $\cH_2$ with Hamiltonians $H_1$ and $H_2$, respectively. However, in this case we do not measure the error by the diamond norm, but by the energy-constrained diamond norm:

\begin{definition}[$\epsilon\text{-EPQP}_\cC$]
\label{epsilonEPQP}
Let $\cH_1$ and $\cH_2$ be separable Hilbert spaces and consider a class of quantum channels $\cC \subset \text{CPTP}(\cH_1,\cH_2)$. 
A quantum operation $\cP \in \text{CPTP} (\cH_1 \otimes \cH_P,\cH_2)$ is called an \emph{$\epsilon$-approximate energy-constrained programmable quantum processor for $\cC$ ($\epsilon\text{-EPQP}_\cC$)} if for all $\Phi \in \cC$ there exists a state $\pi_\Phi \in \cD(\cH_P)$ such that
\begin{equation}
\label{eq:epsilonEPQP}
\frac12 \| \cP(\cdot \otimes \pi_\Phi) - \Phi \|_{\diamond}^{E} \leq \epsilon.
\end{equation}
\end{definition}

We denote the dimension of the program register of an $\epsilon\text{-EPQP}_\cC$ as $d_P^\infty$. We consider only classes $\cC$ of cptp maps that are $(\alpha,\beta)$-energy-limited for some $\alpha\geq 1$, $\beta \geq 0$. Important classes treated next are the following: the set of all  $(\alpha,\beta)$-energy-limited quantum channels from $\cH_1$ to $\cH_2$ is denoted $\cL(\alpha, \beta)$. For $\cH_1=\cH_2$, the set of $(\alpha,\beta)$-energy-limited unitary channels is denoted $\cU(\alpha,\beta)$. In a Bosonic system of a single Bosonic mode, we look at the set $\cG\cC\cG(\alpha,\beta)$ of all $(\alpha, \beta)$-energy-limited gauge-covariant Gaussian channels; finally, for a general (multimode) Bosonic system, we denote the set of all $(\alpha, \beta)$-energy-limited Gaussian unitary channels as $\cG\cU(\alpha, \beta)$.

While the No-Programming Theorem \ref{NoPro} is not directly applicable to an $\epsilon\text{-EPQP}_{\cU(\alpha,\beta)}$, we can reduce the case of infinite dimension to the finite-dimensional case, to rule out the existence of perfect universal programmable quantum processors, by considering only unitaries of the form $U=U_0+\sum_{n=2}^\infty P_n$, where $U_0$ is an arbitrary unitary on the support $\cH_1 \subset \cH$ of $P_0+P_1$, i.e., $U_0 U_0^\ast = U_0^\ast U_0 = P_0+P_1$, which are $(\alpha,\beta)$-energy-limited for all $\alpha\geq 1$, $\beta \geq \epsilon_1-\epsilon_0$. Namely, a $0$-EPQP would imply a $0$-UPQP for input register $\cH_1$, and since the latter is a finite-dimensional space, Theorem \ref{NoPro} applies. 

Before launching into the rest of the paper, where we derive results about various $\epsilon\text{-EPQP}_\cC$, it is worth pausing to regard the definition, and to justify why we use the energy-constrained diamond norm and why we restrict our attention to $(\alpha,\beta)$-energy-limited channels. In fact, it turns out that dropping either restriction results in simple no-go theorems ruling out finite-dimensional program registers.

\begin{remark}[Unsuitability of the usual diamond norm]
\label{rem:no-diamond}
Fix a direct sum decomposition of a closed subspace $\cH_0$ of $\cH$ into orthogonal subspaces, $\cH_0 = \bigoplus_{n}^\infty \cH_n$, where $\dim \cH_n = n$, and consider $\cC = \bigcup_n \cU(\cH_n) \subset \cU(\cH)$, where we regard $\cU(\cH_n)$ as a subset of $\cU(\cH)$ by identifying each unitary $U\in\cU(\cH_n)$ with $U \oplus (\1-UU^*) \in \cU(\cH)$.
An $\epsilon\text{-PQP}_\cC$ would effectively serve as an $\epsilon\text{-PQP}_{\cU(\cH_n)}$ for the unitaries on each of the finite-dimensional spaces $\cH_n$. Previous results~\cite{Majenz, Kubicki19, Renner20}, then show lower bounds on $d_P$ that diverge as a function of $n$, showing that $d_P^\infty$ must be infinite. The energy-constrained diamond norm, besides inducing the more natural strong topology on the cptp maps, avoids this problem. 
Note that if the subspaces $\cH_n$ are contained in eigenspaces of the Hamiltonian $H\geq 0$, the class $\cC$ even consist of $(1,0)$-energy-limited unitary channels.

The same argument also shows why we have to use a \emph{finitary} Hamiltonian for the $E$-constrained diamond norm. Indeed, for a nonfinitary Hamiltonian $H\geq 0$ there exists an energy $E$ such that the subspace $\cH' \subset \cH$ spanned by eigenstates of energy $\leq E$ has infinite dimension. Then, the class of unitaries $\cC=\cU(\cH') \subset \cU(\cH)$ is $(1,E)$-energy-limited, and restricted to $\cH'$, the $E$-constrained diamond norm equals the unconstrained diamond norm. This includes all Hamiltonians with continuous or partially continuous spectrum.
\end{remark}

\begin{remark}[Triviality of energy-unlimited quantum channels]
\label{rem:no-unlimited}
The considerations in the preceding remark should have convinced us that we need to consider Definition \ref{epsilonEPQP} with a finitary Hamiltonian $H_1$ on $\cH_1$. Now let us consider how allowing energy-unlimited channels would similarly trivialize our quest for finite program registers. 
Namely, for every eigenstate $\ket{\phi_n}$ of $H_2$ with energy $e_n$ there exists a well-defined cptp map $\Phi_n$ that maps all of $\cD(\cH_1)$ to $\proj{\phi_n}$. A processor that can $\epsilon$-implement any of the $\Phi_n$, even with respect to $\|\cdot\|_\diamond^{E}$, can approximately prepare an arbitrary eigenstate $\proj{\phi_n} \in\cD(\cH_2)$, with respect to the trace norm, from the ground state of $\cH_1$, so it is intuitively clear that it requires an infinite-dimensional program register. This can be made precise using the information-theoretic lower bound method explained at the end of Section \ref{sec:recycling}.
\end{remark}

\section{Dimension bounds for energy-limited \\programmable quantum processors}
\mylabel{section2}

In the following, we focus on the resources the processor requires to approximately implement all $(\alpha, \beta)$-energy-limited unitary channels $\cU(\alpha, \beta)$. 
To obtain upper bounds on the dimension of the program register in Subsection~\ref{subsec:upperbounds}, we present a construction method based on an existing $\epsilon\text{-PQP}_\cU$. This can be seen as an extension of a finite-dimensional programmable quantum processor to infinite dimension. 

\medskip
A technical lemma we use in the proof is the gentle operator lemma, which shows that a measurement with a highly likely outcome can be performed with little disturbance to the measured quantum state.

\begin{lemma}[{Gentle operator~\cite[Lemma 9]{Winter99},~\cite[Lemma 5]{Ogawa02},~\cite[Lemma~9.4.2]{Wilde13}}] 
\mylabel{gentleoperator}
Let $\rho \in \cD(\cH)$ and $T$ be a measurement operator with $0 \leq T \leq \mathds{1}$.  Suppose that $T$ has a high probability of detecting $\rho$, i.e., $\tr \rho T \geq 1 - \kappa$, with $\kappa \in [0,1]$. Then,
\begin{equation}
  \left\| \rho - \sqrt{T} \rho \sqrt{T} \right\|_1 \leq 2 \sqrt{\kappa} .
\end{equation}
\end{lemma}

In the following theorem, we construct a processor that maps any input state $\rho \in \cD(\cH)$ with a certain energy 
\begin{equation}
\label{Energyassumptioninputstate}
\tr \rho H \leq E,
\end{equation}
approximately to $U \rho U^\ast$, if $U \in \cU(\cH)$ is $(\alpha,\beta)$-energy-limited, using a program register $\cH_P$ of dimension $\dim \cH_P = d_P^\infty < \infty$.  We express the approximation parameter $\gamma$ of our $\gamma\text{-EPQP}_{\cU(\alpha, \beta)}$ in terms of the approximation parameter $\epsilon$ of a given finite-dimensional $\epsilon\text{-PQP}_{\cU}$. This is illustrated in Figure~\ref{fig:expansion}.

\begin{figure}[ht]
\centering
\includegraphics[width=0.6\textwidth]{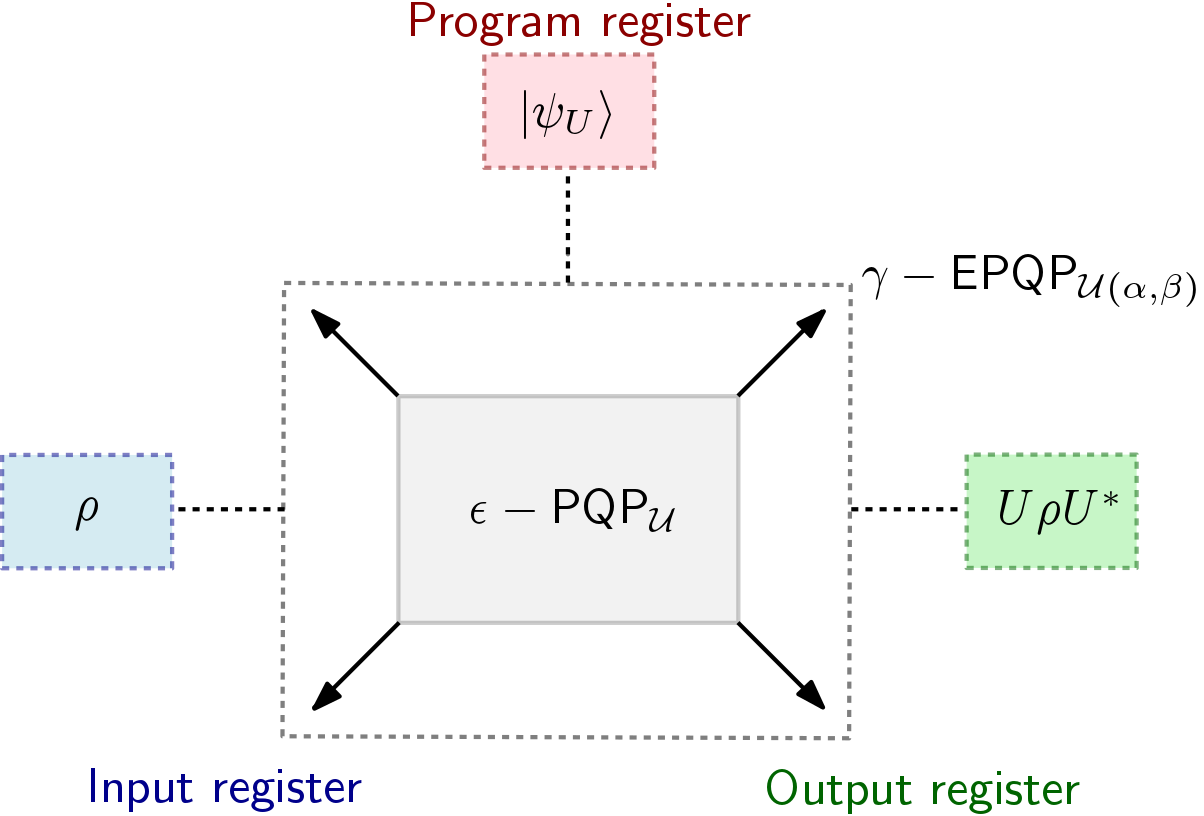}\caption{This is a schematic illustration of the construction method for an $\gamma\text{-EPQP}_{\cU(\alpha, \beta)}$ based on an $\epsilon\text{-PQP}_\cU$. The parts of the figure that we construct in the proof are represented by dashed lines, the parts we assume to exist by regular lines.}\label{fig:expansion}
\end{figure}

To obtain such a processor, we use the energy constraints to approximate the input and output system by a finite-dimensional subspace of $\cH$.

\begin{theorem}
\mylabel{maintheorem}
Let $H \geq 0$ be a finitary and grounded Hamiltonian on the separable Hilbert space $\cH$, and $E >0$. Furthermore, let $\epsilon>0$ and $d \assign \operatorname{rank}\{H\leq E/\epsilon^4\}$, the dimension of the subspace of $\cH$ spanned by eigenvectors of $H$ with eigenvalues $\leq\frac{E}{\epsilon^4}$. 
Assume that we have an $\epsilon\text{-PQP}_\cU$ $\cP_d$ with $d$-dimensional input register and program register $\cH_P$. Then, we can construct an infinite-dimensional $\gamma\text{-EPQP}_{\cU(\alpha, \beta)}$ $\cP \in CPTP (\cH \otimes \cH_P,\cH)$ such that for all $(\alpha,\beta)$-energy-limited unitaries $U \in \cU(\cH)$ there exists a unit vector $\ket{\psi_U} \in \cH_P$ such that
\begin{equation}
\mylabel{ThmEPQP}
\frac{1}{2}\| \cP (\cdot \otimes \proj{\psi_U}) - U (\cdot) U^\ast \|_\diamond^{E} \leq \gamma \assign 4.5 \epsilon\left(\alpha +\frac{\beta}{E}\right).
\end{equation}
\end{theorem}

\begin{proof}
The construction of the infinite-dimensional processor consists of two components: a compression map that projects down to states on a finite-dimensional subspace $\cH_d$ of $\cH$ spanned by the lowest-lying energy eigenstates, and the application of the given finite-dimensional $\epsilon\text{-PQP}_\cU$ $\cP_d$ to that subspace.

Define $P_\delta$, the projector onto the subspace $\cH_\delta$ spanned by all eigenstates with eigenvalue $\leq E/\delta$, where $\delta \leq 1$. Consider $\cH_{\delta^2}$, which has projector $P_{\delta^2}$ and define the compression map $\cK$ onto $\cH_\delta \subset \cH_{\delta^2}$ as $\cK(\rho) := P_\delta \rho P_\delta + \tr\rho(\1-P_\delta) \proj{0}$, where $\ket{0}$ is a ground state of the Hamiltonian $H$.
Now we can define our infinite-dimensional processor as $\cP = \cP_d\circ(\cK\otimes\id_P)$.

Next, we need to describe how to use the processor $\cP$ to implement an $(\alpha,\beta)$-energy-limited unitary $U\in\cU(\cH)$, namely what is the program state $\ket{\psi_U}$. To do this, consider the polar decomposition of $P_{\delta^2} U P_\delta$, which we can think of as an operator acting on $\cH_{\delta}$, mapping to $\cH_{\delta^2}$:
\begin{equation}
  \label{eq:U-U_d}
  P_{\delta^2} U P_\delta
    = V_d \sqrt{P_\delta U^\ast P_{\delta^2} U P_\delta},
\end{equation}
where $V_d:\cH_{\delta} \to \cH_{\delta^2}$ is consequently an isometry. We obtain $U_d$ as an extension of $V_d$ to a unitary on $\cH_d := \cH_{\delta^2}$. By assumption, $\cP_d$ can implement $U_d$ approximately with error $\leq \epsilon$ in diamond norm, using a certain program state $\ket{\phi_{U_d}}$, and we let $\ket{\psi_U} := \ket{\phi_{U_d}}$.

The rest of the proof is the demonstration that this construction satisfies the claimed approximation quality in $E$-constrained diamond norm. To start with, we show that 
\begin{equation}
  \label{eq:PUP-vs-U_d}
  \left\| V_d - P_{\delta^2} U P_\delta \right\| 
     \leq \delta\left(\alpha+\frac{\beta}{E}\right), 
\end{equation}
both operators in the difference being understood as operators on $\cH_d$, and $\|\cdot\|$ denoting the operator norm. For this, thanks to Eq.~\eqref{eq:U-U_d} it is enough to show 
\begin{equation}
  \label{eq:P_delta-bound}
  \left(1 - \delta\left(\alpha+\frac{\beta}{E}\right) \right)P_\delta \leq P_\delta U^* P_{\delta^2} U P_\delta \leq P_\delta,
\end{equation}
since the square root is operator monotonic, and thus implies
\begin{equation*}
  \left\| P_\delta - \sqrt{P_\delta U^* P_{\delta^2} U P_\delta} \right\| 
     \leq \delta\left(\alpha+\frac{\beta}{E}\right). 
\end{equation*}
The right-hand inequality in Eq.~\eqref{eq:P_delta-bound} follows trivially from $P_{\delta^2} \leq \1$ by conjugation with $U^*$ and $P_\delta$. The left-hand inequality amounts to showing $\bra{\psi} U^* P_{\delta^2} U \ket{\psi} \geq 1 - \delta\left(\alpha+\frac{\beta}{E}\right)$ for all state vectors $\ket{\psi}\in\cH_\delta$. Indeed, $\proj{\psi}$ is supported on the subspace $\cH_\delta$, all of whose state vectors have energy $\leq E/\delta$, in particular $\tr \big( \proj{\psi} H \big ) \leq E/\delta$. By our assumption that $U$ is $(\alpha,\beta)$-energy-limited, this implies $\tr \big( U\proj{\psi}U^* H \big) \leq \alpha E/\delta + \beta$, and thus by Markov's inequality
we get
\[
  \tr \big( U\proj{\psi}U^* (\1-P_{\delta^2}) \big)
      \leq \frac{\alpha E/\delta + \beta}{E/\delta^2}
      \leq \delta\left(\alpha+\frac{\beta}{E}\right),
\]
proving the claim.

To prove the bound from Eq.~\eqref{ThmEPQP}, we consider an arbitrary state $\rho\in\cD(\cH\otimes\mathbb{C}^k)$ with energy bounded by $E$, i.e., $\tr \rho (H \otimes \mathds{1}_{\Ck}) \leq E$. To start, by Markov's inequality this implies $\tr\rho(P_\delta\otimes\1) \geq 1-\delta$, thus by the Gentle Operator Lemma \ref{gentleoperator}, and the triangle inequality,
\begin{equation}
  \label{eq:rho-Krho}
  \bigl\| \rho-(\cK\otimes\id_{\mathbb{C}^k})(\rho) \bigr\|_1 \leq 2\sqrt{\delta}+\delta, 
\end{equation}
and furthermore $(\cK\otimes\id_{\mathbb{C}^k})(\rho)$ is a state on $\cH_d\otimes\mathbb{C}^k$ that has energy bounded by  $\tr(\cK\otimes\id_{\mathbb{C}^k})(\rho)(H \otimes\1) \leq \tr\rho(H \otimes\1) \leq E$.

Now, from the definition of $\cP$ and the processor property of $\cP_d$, we have
\begin{equation}
  \label{eq:Pd-epsilon}
  \bigl\| (\cP\otimes\id_{\mathbb{C}^k})(\rho\otimes\psi_U)
         - (U_d \otimes\1_{\mathbb{C}^k})
           (\cK\otimes\id_{\mathbb{C}^k})(\rho) 
           (U_d \otimes\1_{\mathbb{C}^k})^* \bigr\|_1 \leq 2\epsilon. 
\end{equation}
Noting that $(U_d \otimes\1_{\mathbb{C}^k}) (\cK\otimes\id_{\mathbb{C}^k})(\rho) (U_d \otimes\1_{\mathbb{C}^k})^* = (V_d \otimes\1_{\mathbb{C}^k}) (\cK\otimes\id_{\mathbb{C}^k})(\rho) (V_d \otimes\1_{\mathbb{C}^k})^*$, because $(\cK\otimes\id_{\mathbb{C}^k})(\rho)$ is supported on $\cH_\delta\otimes\mathbb{C}^k$, we furthermore have
\begin{equation}\begin{split}
  \label{eq:V-vs-PUP-with-K}
  &\bigl\| (V_d \otimes\1_{\mathbb{C}^k})
          (\cK\otimes\id_{\mathbb{C}^k})(\rho) 
          (V_d \otimes\1_{\mathbb{C}^k})^* \bigr. \\
  &\phantom{==}
   \bigl. - (P_{\delta^2}UP_\delta \otimes\1_{\mathbb{C}^k})
            (\cK\otimes\id_{\mathbb{C}^k})(\rho) 
            (P_\delta U^*P_{\delta^2} \otimes\1_{\mathbb{C}^k})
            \bigr\|_1 
  \leq 2\delta\left(\alpha+\frac{\beta}{E}\right),
\end{split}\end{equation}
where we invoke Eq.~\eqref{eq:PUP-vs-U_d} twice. 
Continuing, we observe that we can drop the projection $P_\delta$ in the second term inside the norm, because $(\cK\otimes\id_{\mathbb{C}^k})(\rho)$ is supported on $\cH_\delta\otimes\mathbb{C}^k$. Next, by Eq.~\eqref{eq:rho-Krho} we have
\begin{equation}
  \label{eq:rho-Krho-bis}
  \bigl\| (P_{\delta^2}U \otimes\1_{\mathbb{C}^k})
            (\cK\otimes\id_{\mathbb{C}^k})(\rho) 
            (U^*P_{\delta^2} \otimes\1_{\mathbb{C}^k})
          - (P_{\delta^2}U \otimes\1_{\mathbb{C}^k})
            \rho
            (U^*P_{\delta^2} \otimes\1_{\mathbb{C}^k})
            \bigr\|_1 
  \leq 2\sqrt{\delta}+\delta.
\end{equation}
Finally, since $\tr \big((U\otimes\1_{\mathbb{C}^k})\rho(U\otimes\1_{\mathbb{C}^k})^* H \big) \leq \alpha E+\beta$, another application of Markov's inequality and the Gentle Operator Lemma \ref{gentleoperator} yields
\begin{equation}
  \label{eq:drop-P_delta2}
  \bigl\| (P_{\delta^2}U \otimes\1_{\mathbb{C}^k})
          \rho
          (U^*P_{\delta^2} \otimes\1_{\mathbb{C}^k})
           - (U \otimes\1_{\mathbb{C}^k})
             \rho
             (U^*\otimes\1_{\mathbb{C}^k}) \bigr\|_1 
    \leq 2\delta\sqrt{\alpha+\frac{\beta}{E}}.
\end{equation}
It remains to put everything together: by the triangle inequality and the bounds from Eqs.~\eqref{eq:Pd-epsilon}, \eqref{eq:V-vs-PUP-with-K}, \eqref{eq:rho-Krho-bis} and \eqref{eq:drop-P_delta2}, we obtain
\[\begin{split}
  \bigl\| (\cP\otimes\operatorname{id}_{\mathbb{C}^k})(\rho\otimes\psi_U)
          &- (U \otimes\mathds{1}_{\mathbb{C}^k}) \rho
            (U^*\otimes\mathds{1}_{\mathbb{C}^k}) \bigr\|_1 \\
    &\leq 2\epsilon + 2\delta\left(\alpha+\frac{\beta}{E}\right) + 2\sqrt{\delta}+\delta + 2\delta\sqrt{\alpha+\frac{\beta}{E}} \\
   & \leq 2\epsilon+7\sqrt{\delta}\left(\alpha+\frac{\beta}{E}\right),
\end{split}\]
and choosing $\delta=\epsilon^2$ concludes the proof.
\end{proof}

To obtain lower bounds on the dimension of the program register in Subsection~\ref{subsec:lowerbounds}, in the following theorem we present a method to construct a finite-dimensional processor assuming an existing infinite-dimensional one, which is illustrated in Figure~\ref{fig:contraction}.

\begin{figure}[ht]
\centering
\includegraphics[width=0.6\textwidth]{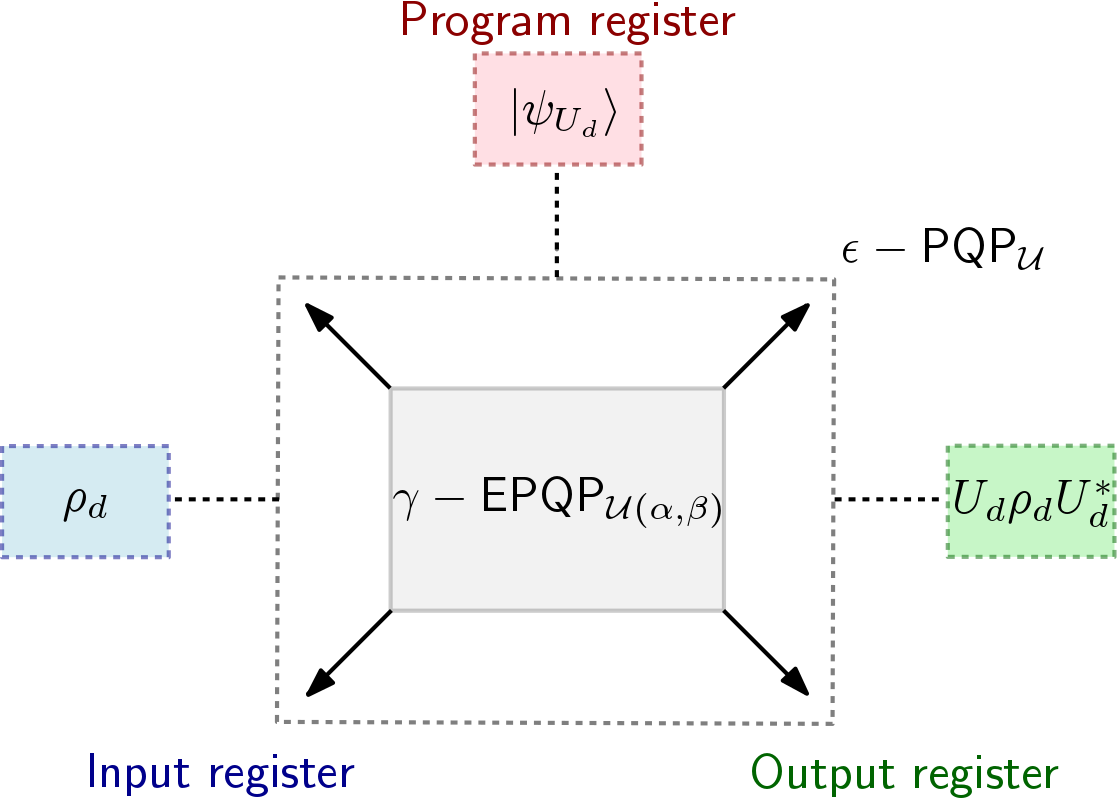}
\caption{Assuming a $\gamma\text{-EPQP}_{\cU(\alpha, \beta)}$, we can construct a finite-dimensional $\epsilon\text{-PQP}_{\cU}$ which is drawn in dashed lines including its input, output and program register.}\label{fig:contraction}
\end{figure}

\begin{theorem}
\label{thm:constructionforlowerbounds}
Let $H\geq 0 $ be the Hamiltonian with a discrete spectrum describing the system on the separable Hilbert space $\cH$, $E>0$, $\gamma>0$, and furthermore choose $d>0$. Assume that we have an infinite-dimensional $\gamma\text{-EPQP}_{\cU(\alpha, \beta)}$ $\cP \in CPTP(\cH \otimes \cH_P, \cH)$ for all sufficiently large $\alpha$ and $\beta$. Then, there exists an $\epsilon\text{-PQP}_\cU$ $\cP_d \in CPTP (\cH_d \otimes \cH_P, \cH_d)$ such that for all $U_d \in \cU(\cH_d)$ there is a unit vector $\ket{\psi_{U_d}} \in \cH_P$ with
\begin{equation}
    \frac{1}{2}\left\| \cP_d (\cdot \otimes \proj{\psi_{U_d}}) - U_d (\cdot) U_d^* \right\|_\diamond \leq \epsilon,
\end{equation}
where $\epsilon = \gamma\frac{1}{E}\max\{E(d),E\}$ and $E(d)$ is the smallest energy such that the space spanned by the eigenstates of energies between $0$ and $E(d)$ is of dimension $d$ or larger.
\end{theorem}

\begin{proof}
We assume that an infinite-dimensional  $\gamma\text{-EPQP}_{\cU(\alpha, \beta)}$ processor $\cP \in CPTP (\cH \otimes \cH_P, \cH)$ exists as described in the theorem, and construct a finite-dimensional one with the same program register, i.e., we aim to bound
\begin{equation}
\frac{1}{2}\|\cP_d (\cdot \otimes \ket {\psi_{U_d}} \bra {\psi_{U_d}}) - U_d (\cdot) U_d^\ast \|_\diamond \leq \epsilon := \epsilon (E, d, \gamma).
\end{equation} 
We start with fixing an isometric embedding $V$ of the $d$-dimensional Hilbert space $\cH_d$ into $\cH$. Namely, with respect to the ordered spectral decomposition $H=\sum_{n=0}^\infty e_n P_n$ of the Hamiltonian, let $n(d)$ be the smallest integer such that $d \leq \sum_{n=0}^{n(d)} \operatorname{rank} P_n$ and $E(d) := e_{n(d)}$ the largest occurring energy. 
Let 
\begin{equation}
    V:\cH_d \hookrightarrow \left(\sum_{n_0}^{n(d)} P_n \right)\cH =: \cH' \subset \cH, 
\end{equation}
where the first embedding is an arbitrary isometry. This defines an isometric channel 
\begin{equation}
\begin{split}
  \cV: \cD(\cH_d) &\to \cD(\cH) \\
             \rho &\mapsto V\rho V^*.
\end{split}
\end{equation}
Thanks to $V$ we view $\cH_d$ as a subspace of $\cH$, and 
denote its orthogonal complement by $\cH_d^\perp$. 

An arbitrary $U_d \in\cU(\cH_d)$ is extended to a unitary
$U \equiv U_d \oplus \1_{\cH_d^\perp} \in \cU(\cH)$. 
We would like to implement this unitary using the processor $\cP$, followed by a cptp compression onto the subspace $\cH_d$, using its projection operator $\Pi_d$:
\[
  \cK(\rho) \assign \Pi_d \rho \Pi_d + \kappa_0 \tr\rho(\1-\Pi_d),
\]
where $\kappa_0$ is an arbitrary state with energy zero in $\cH_d$. Since we assume that there exists a $\gamma\text{-EPQP}_{\cU(\alpha, \beta)}$ for all sufficiently large $\alpha >0 $ and $\beta>0$, this unitary can be $\gamma$-implemented. (In fact, we could choose $\alpha=1$ and $\beta = E(d)$, the largest occurring energy gap in $\cH'$.)
The processor is a concatenation of the isometric channel $\cV$, the infinite-dimensional $\gamma\text{-EPQP}_{\cU(\alpha, \beta)}$ $\cP$ and the compression map $\cK$, 
namely $\cP' \assign \cK\circ\cP\circ(\cV\otimes\id_{\cH_P})$, which leads us to
\begin{equation}\begin{split}
    \frac{1}{2}\| \cK \circ \cP \circ (\cV (\cdot) \otimes \psi_{U_d}) - U_d (\cdot)U_d^\ast \|_\diamond 
    &\leq \frac{1}{2}\| \cP\circ (\cV (\cdot)\otimes \psi_{U_d}) - U (\cdot)U^\ast \|_{\diamond}^{E(d)} \\
    &\leq \frac{1}{E}\max\{E(d),E\} \frac{1}{2}\| \cP\circ (\cV (\cdot)\otimes \psi_{U_d}) - U (\cdot)U^\ast \|_{\diamond}^E  \\
    &\leq \gamma \frac{1}{E}\max\{E(d),E\},  
\end{split}\end{equation}
where we use the contractivity of the diamond norm under postprocessing and that our subspace goes up to energy $E(d)$, so restricted to it the $E(d)$-constrained diamond norm equals the unconstrained diamond norm; then, that going to the $E$-constrained diamond norm blows up the error by a factor of at most $\frac{1}{E}\max\{E(d),E\}$; finally, the infinite-dimensional processor makes an error of at most $\gamma$. 
Hence, this is the $\epsilon = \gamma\frac{1}{E}\max\{E(d),E\}$ we get for the resulting finite-dimensional processor. 
\end{proof}

\subsection{Upper bounds}
\label{subsec:upperbounds}
Upper bounds on the program dimension of a finite-dimensional $\epsilon\text{-PQP}_\cU$ were derived in various previous works, most recently in Refs.~\cite{Kubicki19} and~\cite{Renner20}. 
To specify an upper bound for the dimension of the program of the infinite-dimensional $\gamma\text{-EPQP}_{\cU(\alpha, \beta)}$ (Definition~\ref{epsilonEPQP}), we thus import the existing bounds via Theorem~\ref{maintheorem}.

\begin{table}[ht]
\centering
\renewcommand{\arraystretch}{1.5}
\begin{tabular}[c]{| l | l | l |}
 \hline
 $d_P\leq$ & References & $d_P^\infty \leq$ \\
 \hline\hline
 \phantom{$d_P\leq$}$\displaystyle{\left( \frac{d}{\epsilon^2} \right)^{4d^2}}$\vspace{1mm} & \parbox{5cm}{Ishizaka \&{} Hiroshima~\cite{Port-based08}, \\ Beigi \&{} K\"onig~\cite{Beigi11}, \\ Christandl \emph{et al.}~\cite{Christiandl18}} &
 \phantom{$d_P^\infty\leq$}$\displaystyle{\left( \frac{20.25 d (\alpha + \frac{\beta}{E})^2}{\gamma^2} \right)^{4d^2}}$ \\ \hline
  \phantom{$d_P\leq$()}$\displaystyle{d^{\frac{2 d^2}{\epsilon}}}$\vspace{1mm} & Pirandola \emph{et al.}~\cite{Pirandola19}\mylabel{PirandolaUpper} & \phantom{$d_P^\infty\leq$()}$\displaystyle{d^{\frac{9d^2  (\alpha + \beta / E )}{\gamma}}} $ \\ \hline
  \phantom{$d_P\leq$}$\displaystyle{\left( \frac{\tilde C}{\epsilon} \right) ^{d^2}}$\vspace{1mm} & Kubicki \emph{et al.}~\cite{Kubicki19}\mylabel{KubickiUpper} & \phantom{$d_P^\infty\leq$}$\displaystyle{\left( \frac{4.5 \tilde C (\alpha + \frac{\beta}{E})}{\gamma} \right) ^{d^2}}$ \\ \hline
 \phantom{$d_P\leq$}$\displaystyle{\left( \frac{\tilde C d^2}{\epsilon}\right)^{\frac{d^2-1}{2}}}$\vspace{1mm} & Yang \emph{et al.}~\cite{Renner20} & \phantom{$d_P^\infty\leq$} $\displaystyle{\left( \frac{4.5 \tilde C d^2 (\alpha + \frac{\beta}{E})}{\gamma}\right)^{\frac{d^2-1}{2}}}$ \\
 \hline
\end{tabular}
\caption{Upper bounds on the program dimension for infinite-dimensional processors from bounds for finite-dimensional processors. Note that $d$ in the third column is the dimension of the subspace of $\cH$ spanned by eigenvectors of $H$ with eigenvalues $\leq\frac{E}{\epsilon^4}$. }
\label{tab:finiteupperbounds} 
\end{table}

Let $\cP$ be an infinite-dimensional $\gamma\text{-EPQP}_{\cU(\alpha, \beta)}$ as in Theorem~\ref{maintheorem} and $\cP_d$ a finite-dimensional $\epsilon\text{-PQP}_\cU$ as required. Since our construction of an infinite-dimensional processor relies on a finite-dimensional one, we reformulate the bound from Eq.~\eqref{ThmEPQP} as
\begin{equation}
\epsilon \assign \frac{\gamma}{4.5}\left(\alpha+\frac{\beta}{E}\right)^{-1}.
\end{equation}

Several upper bounds on the program dimension for finite-dimensional unitary processors can be found in the literature. Note that the bound in the second row of Table~\ref{tab:finiteupperbounds} was derived from Ref.~\cite[Lemma 1, Section II.C.]{Pirandola19} which uses port-based teleportation working with copies of Choi-states. 
We get upper bounds for our infinite-dimensional $\gamma\text{-EPQP}_{\cU(\alpha \beta)}$ $\cP$ if we insert $\epsilon$ into the existing bounds in Table~\ref{tab:finiteupperbounds}.

\subsection{Lower bounds}
\label{subsec:lowerbounds}

Given an $\epsilon\text{-EPQP}_{\cU(\alpha,\beta)}$ for infinite-dimensional $(\alpha, \beta)$-energy-limited unitaries, we can build a finite-dimensional $\epsilon\text{-PQP}_\cU$ via Theorem~\ref{thm:constructionforlowerbounds}, whose program dimension is lower bounded through results from the literature. 
The following lower bounds for $d$-dimensional $\epsilon\text{-PQP}_\cU$ are known, as shown in Table \ref{tab:finitelowerbounds}.

\begin{table}[ht]
\centering
\renewcommand{\arraystretch}{1.5}
\begin{tabular}[c]{| l | l | l |}
 \hline
 $d_P\geq$ & References & $d_P^\infty \geq$ \\
 \hline\hline
 \phantom{$d_P:$} $\displaystyle{K \left(\frac{1}{d}\right)^{\!\frac{d+1}{2}} \!\!\left(\frac{1}{\epsilon}\right)^{\!\frac{d-1}{2}}}$\vspace{1mm} & P\'erez-Garc\'ia~\cite{Perez-Garcia06} & \phantom{$d_P^\infty:$} $\displaystyle{K \left(\frac{1}{d}\right)^{\frac{d+1}{2}} \!\!\left(\frac{E}{\gamma\max\{E(d),E\}}\right)^{\frac{d-1}{2}}}$\\ \hline
 \phantom{$d_P:$} $\displaystyle{\left(\frac{d}{\epsilon} \right)^2}$\vspace{1mm} & Majenz~\cite{Majenz} & \phantom{$d_P^\infty:$} $\displaystyle{\left(\frac{d E}{\gamma\max\{E(d),E\}} \right)^2}$\\ \hline
 \phantom{$d_P:$} $\displaystyle{2^{\frac{1-\epsilon}{3C}d - \frac{2}{3} \log d}}$\vspace{1mm} & Kubicki \emph{et al.}~\cite{Kubicki19} &  \phantom{$d_P^\infty:$} $\displaystyle{2^{\frac{E-\gamma\max\{E(d),E\}}{3CE}d - \frac{2}{3} \log d}}$ \\ \hline
 \phantom{$d_P:$} $\displaystyle{\left(1+\frac{\Theta(d^{-2})}{\sqrt\epsilon} \right)^{2\alpha}}$ \vspace{1mm} & \parbox{3.5cm}{Yang \emph{et al.}~\cite{Renner20}, with\\ a slight arithmetic\\ improvement} & \phantom{$d_P^\infty:$} $\displaystyle{\left(1+\frac{\Theta(Ed^{-2})}{\sqrt{\gamma}\max\{E(d),E\}} \right)^{2\alpha}}$ \\
 \hline
\end{tabular}
\caption{Lower bounds on the program dimension for infinite-dimensional processors from bounds for finite-dimensional processors. The last row holds for any $\alpha < \frac{d^2-1}{2}$. The dimension $d$ in the third column is the chosen dimension of the finite-dimensional processor.}
\label{tab:finitelowerbounds}
\end{table}

\section{Recycling the program for implementing unitaries}
\label{sec:recycling}

In this section, we prove a general lemma which was first shown by Yang \emph{et al.}~\cite{Renner20} for finite-dimensional unitaries and using the diamond norm. We generalize their statement to infinite-dimensional systems. This lemma is applicable to any unitary that a given processor implements approximately, and it says that in such a case the processor can be modified to one that reuses the program state several times to approximately implement the same unitary several times sequentially or in parallel. This statement is crucial in the information-theoretic lower bounds obtained in Ref.~\cite{Renner20} for the program dimension of a universal PQP in finite dimension. We use our infinite-dimensional, energy-constrained version in the next section to give lower bounds on the program dimension of approximate EPQP's for Gaussian unitaries. 

Before we state and prove the lemma, we recall some definitions from Shirokov~\cite{Shirokov16}. 
The \emph{completely bounded energy-constrained channel fidelity} between two quantum channels is defined as
\begin{equation}
F_{cb}^E (\cA, \cB) \assign \inf_{\psi} F \bigl( ( \cA \otimes \id_R )(\psi), (\cB \otimes \id_R )(\psi) \bigr),
\end{equation}
such that $\ket \psi$ varies over states on $\cH_1 \otimes\cH_R$ with energy constraint $\tr \rho H_1 \leq E$. We take, without loss of generality, the infimum over all pure states $\ket \psi \in \cH_1 \otimes \cH_R$ with $\cH_R \simeq \cH_1$ being a reference system, and $F(\rho, \sigma) \assign \tr \sqrt{\rho^{\frac{1}{2}} \sigma \rho^{\frac{1}{2}}}$ for $\rho, \sigma \in \cD(\cH_1)$ denoting the usual fidelity.

From Ref.~\cite{Shirokov16} we have
\begin{equation}
  \label{eq:fidelitydiamond}
  F_{cb}^E(\cA,\cB) \geq 1-\frac12 \|\cA - \cB\|_\diamond^E 
\end{equation} and 
\begin{equation}
  \label{eq:diamondfidelity}
  \|\cA -\cB \|_\diamond^E \leq 2 \sqrt{1-F_{cb}^E(\cA,\cB)^2}.
\end{equation}

Finally, for an $\ell$-partite system with Hilbert space $\cH^{(\ell)} = \cH_1\otimes\cH_2\otimes\cdots\otimes\cH_\ell$, 
where each $\cH_j$ carries its own Hamiltonian $H_j \geq 0$, and for numbers $E_j>0$, we define the \emph{multiply energy-constrained diamond norm}, or more precisely the $(E_1,\ldots,E_\ell)$-constrained diamond norm, of a Hermitian-preserving superoperator $\Phi:\cT(\cH^{(\ell)})\rightarrow \cT(\cH')$ as
\begin{equation}
  \label{eq:multiply-E-diamond}
  \|\Phi\|_\diamond^{(E_1,\ldots,E_\ell)} 
    \assign \sup_{\rho\in\cD(\cH^{(\ell)}\otimes\cH_R)} \|(\Phi\otimes\id_R)(\rho)\|_1
    \quad\text{such that } \forall j=1,\ldots,\ell\ \tr\rho_j H_j \leq E_j .
\end{equation}
This is evidently a norm, being in fact equivalent to the energy-constrained diamond norm: concretely, with $E_{\min}=\min_j E_j$ and $E_{\text{sum}}=\sum_j E_j$,
\begin{equation}
  \|\Phi\|_\diamond^{E_{\min}} 
    \leq \|\Phi\|_\diamond^{(E_1,\ldots,E_\ell)} 
      \leq \|\Phi\|_\diamond^{E_{\text{sum}}},
\end{equation}
where we use $H = \sum_j H_j$ as the Hamiltonian on $\cH^{(\ell)}$.
As with the energy-constrained diamond norm, the induced topology is not the issue, but the fact that the multiple constraints in Eq.~\eqref{eq:multiply-E-diamond} allow us to encode more refined metric information.

\begin{lemma}[{Replication of infinite-dimensional unitaries}]
\label{lemma:E-Yang}
Consider a processor, i.e., a quantum channel $\cP \in \text{CPTP}(\cH \otimes \cH_P, \cH)$, coming with a Hamiltonian $H\geq 0$ and a number $E>0$, and an integer $\ell\geq 1$. Then, there exists another processor $\widehat{\mathcal{P}} \in \text{CPTP}(\mathcal{H}^{\otimes \ell} \otimes \mathcal{H}_P, \mathcal{H}^{\otimes \ell})$ with the following property: 
for every unitary channel $\mathcal{U}(\cdot)= U\cdot U^*$ whose inverse $\mathcal{U}^*$ is $(\alpha',\beta')$-energy-limited and such that there exists a pure state $|{\psi_U}\rangle$ on $\mathcal{H}_P$ with
\begin{equation}
  \label{eq:single-U}
  \frac{1}{2} \| \mathcal{P}(\cdot \otimes \psi_U) - \mathcal{U} \|_\diamond^E \leq \epsilon,
\end{equation}
it holds that 
\begin{equation}
    \label{eq:muliply-U}
    \frac{1}{2} \left\| \widehat{\mathcal{P}}(\cdot \otimes \psi_U) - \mathcal{U}^{\otimes \ell} \right\|_\diamond^{(E,\ldots,E)} \leq 2\ell\epsilon',
\end{equation}
where $\epsilon' = \left(1+\frac{\beta'}{E}\right)\sqrt{2\epsilon}$.
\end{lemma}

In words, whenever $\cP$, using a pure program state $\psi_U$, $\epsilon$-implements an $(\alpha,\beta)$-energy-limited unitary channel $\cU(\cdot)= U \cdot U^*$ with $(\alpha',\beta')$-energy limited inverse $\cU^*$ (with respect to the $E$-constrained diamond norm),
then $\widehat{\cP}$, using the same pure program state $\psi_U$, $2\ell\epsilon'$-implements $\cU^{\otimes\ell}$ (with respect to the $(E,\ldots,E)$-constrained diamond norm). In the finite-dimensional setting and without energy constraint, the above statement reduces to that of Ref.~\cite{Renner20}.

\begin{proof}
Applying Eq.~\eqref{eq:fidelitydiamond}, we obtain for the energy-constrained completely bounded fidelity between the output of the processor and the target unitary, 
\begin{equation}
    F_{cb}^E\bigl(\cP(\cdot\otimes\psi_P), \cU\bigr) \geq 1- \epsilon. 
\end{equation}
Throughout the proof, we observe the convention to mark states for clarity with the index of the subsystem on which they act; similarly for channels, whenever it is not clear from the setting. 
Let $\cV: \cB(\cH) \otimes \cB(\cH_P) \to \cB(\cH) \otimes \cB(\cH_Q)$ be a Stinespring dilation 
of $\cP$, with $\cH_Q$ being a suitable environment space. 
Then, by Uhlmann's theorem for the completely bounded energy-constrained fidelity of quantum channels, stating that any two channels have isometric dilations with the same completely bounded energy-constrained fidelity~\cite[Prop.~1]{Shirokov16} (generalizing the case without energy constraint~\cite{KSW:dilation2}), there exists a state
$\ket{\phi_Q}\in \cH_Q$ such that
\begin{equation}
    F_{cb}^E\bigl(\cV \circ (\id_{\cH} \otimes \psi_P), \cU \otimes \phi_Q\bigr) 
       \geq 1 - \epsilon.
\end{equation}
Note that here and in the following, we regard a state $\rho\in\cD(\cH_P)$ 
as a channel from a trivial system, denoted $1$ (with one-dimensional 
Hilbert space $\mathbb{C}$), to $\cH_P$. 
Using Eq.~\eqref{eq:diamondfidelity}, we get
\begin{equation}
    \label{eq:purified-PQP}
    \frac12 \left\|\cV \circ (\id_{\cH} \otimes \psi_P) - \cU \otimes \phi_Q \right\|_\diamond^E 
          \leq \sqrt{1 - (1-\epsilon)^2} \leq \sqrt{2 \epsilon}.
\end{equation}

We want to apply the processor several times, and for this we have to recover 
the program state. So, the first step of the proof is to show that we can modify the processor to a new map $\cP'\in\text{CPTP}(\cH\otimes\cH_P,\cH\otimes\cH_P)$ in such a way that apart from implementing $\cU$, it also preserves the program state $\psi_P$ (all with the appropriate approximations).
For this purpose, choose a pseudoinverse of $\cV$,
\begin{equation}
    \cW \assign \cV^\ast + \cR = V^*(\cdot)V + \cR,
\end{equation}
where $\cR$ is a cp map designed to make $\cW$ cptp. Note that this 
is always possible because $VV^*$ is a projection, and so $\cV^*$ is
completely positive and trace nonincreasing. In fact, denoting 
$\Pi \assign VV^*$ the projection operator onto the image of $V$ 
in $\cH\otimes\cH_Q$, one choice is $\cR(\xi) = \rho_0 \tr\xi(\1-\Pi)$,
with a fixed state $\rho_0 \in \cD(\cH\otimes\cH_P)$. 
We observe that indeed $\cW\circ\cV = \id_{\cH}\otimes \id_{\cH_P}$,
and hence $(\cV\circ\cW)\circ\cV = \cV$. 
Thus we get from Eq.~\eqref{eq:purified-PQP} that 
\begin{equation}\begin{split}
    \label{eq:recovery-1}
    \sqrt{2\epsilon} 
       &\geq \frac12 \left\|\cV \circ (\id_{\cH} \otimes \psi_P) - \cV\circ\cW\circ(\cU \otimes \phi_Q) \right\|_\diamond^E \\
       &=     \frac12 \left\| \id_{\cH} \otimes \psi_P - \cW\circ(\cU \otimes \phi_Q) \right\|_\diamond^E \\
       &\geq \frac12 \frac{E}{E+\beta'}
             \left\| \id_{\cH} \otimes \psi_P - \cW\circ(\cU \otimes \phi_Q) \right\|_\diamond^{E+\beta'},
\end{split}\end{equation}
where we obtain the second line by the invariance of the (energy-constrained) diamond 
norm under multiplication from the left by an isometric channel ($\cV$),
and the last line by the equivalence of the energy-constrained diamond 
norms for different energy levels~\cite{Winter17}. 

Since $\cU^*$ is $(\alpha',\beta')$-energy-limited, we can lower bound 
the last expression in turn by
\begin{equation}
    \frac12 \frac{E}{E+\beta'}
             \left\| \cU^* \otimes \psi_P - \cW\circ(\id_{\cH} \otimes \phi_Q) \right\|_\diamond^{E/\alpha'}. 
\end{equation}
Defining the memory recovery map $\cM:\cT(\cH_Q)\rightarrow\cT(\cH_P)$ by 
\begin{equation}
    \cM(\rho) \assign \tr_{\cH} \cW(\proj{0}_{\cH}\otimes \rho),
\end{equation}
where $\proj{0}$ is a ground state (i.e., of zero energy) of the 
Hamiltonian $H$, we can now conclude that 
\begin{equation}
    \frac12 \| \cM(\phi_Q) - \psi_P \|_1 
       \leq \left(1+\frac{\beta'}{E}\right)\sqrt{2\epsilon} =: \epsilon'.
\end{equation}
Note that this map depends only on the chosen Stinespring dilation $\cV$
of $\cP$, and thus we can define
\begin{equation}
    \cP' \assign (\id_{\cH}\otimes\cM)\circ\cV,
\end{equation}
which by the above reasoning has the desired property that
\begin{equation}
    \label{eq:processor-modified}
    \frac12 \| \cP'\circ(\id_{\cH}\otimes\psi_P) - \cU\otimes\psi_P\|_\diamond^E
       \leq 2\epsilon',
\end{equation}
via a simple application of the triangle inequality and the contractivity of the (energy-constrained) diamond norm under multiplication from the left by cptp maps. 

At this point we are almost there, and the remaining argument only 
requires careful notation. We have isomorphic copies $\cH_j$ of $\cH$,
$j=1,\ldots,\ell$ as well as the program register $\cH_P$ in the big 
tensor-product space $\cH^{(\ell)}\otimes\cH_P$, and we use 
index $j$ or P to indicate on which tensor factor a superoperator acts
(implicitly extending the action to the whole space by tensoring 
with the identity $\id$ on the other factors). This allows us to 
rewrite Eq.~\eqref{eq:processor-modified} as
\begin{equation}
    \label{eq:processor-modified-2}
    \frac12 \| \cP'_{jP}\circ(\id_j\otimes\psi_P) - \cU_j\otimes\psi_P\|_\diamond^E
       \leq 2\epsilon',
\end{equation}
for all $j=1,\ldots,\ell$. 
By tensoring this with identities $\cU_k$ (for $k<j$) and with $\id_k$ ($k>j$),
and observing the definition of the multiply energy-constrained diamond norm, 
this results in 
\begin{equation}\begin{split}
    \label{eq:processor-multiply}
    \frac12 &\bigl\| \cP'_{jP}\circ(\cU_1\otimes\cdots\otimes\cU_{j-1}\otimes\id_j\otimes\id_{j+1}\otimes\cdots\otimes\id_\ell\otimes\psi_P) \bigr. \\
            &\phantom{====}\bigl.
               - \cU_1\otimes\cdots\otimes\cU_{j-1}\otimes\cU_j\otimes\id_{j+1}\otimes\cdots\otimes\id_\ell\otimes\psi_P \bigr\|_\diamond^{(E,\ldots,E)}
       \leq 2\epsilon',
\end{split}\end{equation}
for $j=1,\ldots,\ell$. 
(It is perhaps worth noting that we impose only energy constraints on the 
nontrivial systems $\cH_j$, both in the simply and multiply constrained 
norm expressions.) 

Adding all the $\ell$ bounds from Eq.~\eqref{eq:processor-multiply}, 
and using the triangle inequality results in 
\begin{equation}\begin{split}
    \frac12 &\bigl\| \cP'_{\ell P}\circ\cdots\circ\cP'_{2P}\circ\cP'_{1P}
                 \circ(\id_1\otimes\cdots\otimes\id_\ell\otimes\psi_P) \bigr. \\
            &\phantom{=============:}\bigl.
               - \cU_1\otimes\cdots\otimes\cU_\ell\otimes\psi_P \bigr\|_\diamond^{(E,\ldots,E)} \leq 2\ell\epsilon'. 
\end{split}\end{equation}
This means that we can define our desired EPQP via
$\widehat{\cP} 
 \assign \tr_P\circ\cP'_{\ell P}\circ\cdots\circ\cP'_{2P}\circ\cP'_{1P}$,
concluding the proof.
\end{proof}

\medskip
We do not make use of it, but from the proof (and much more explicitly from that of Ref. [2]), something stronger can be obtained. Namely, the processor can be modified in such a way that, rather than merely implementing $\ell$ instances of $\mathcal{U}$ sequentially, it does so by alternating them with $\ell$ instances of $\mathcal{U}^\ast$. This is quite curious and hints at an interesting property of programmable quantum processors for unitary channels, namely that with every unitary it approximately implements, it essentially (i.e., after suitable modification) also implements the inverse of the unitary. 

\medskip
The methodology for lower bounding the program register, expounded in Ref.~\cite{Renner20}, is information theoretic, and in itself does not rely on the unitarity of the target channels. Namely, choose a fiducial state $\rho_0\in \cD(\cH)$ of energy $\leq E$ to be used in the channels $\Phi\in\cC$ and the processor, as well as a probability distribution $\mu({\rm d}\Phi)$ on the class $\cC$, so that by Definition \ref{epsilonEPQP} [Eq.~\eqref{eq:epsilonEPQP}] we have for all $\Phi\in\cC$,
\[
  \frac12 \| \cP(\rho_0\otimes\pi_\Phi) - \Phi(\rho_0)\|_1 \leq \epsilon.
\]
Now we can consider three ensembles of states, all sharing the same probability distribution $\mu$:
\begin{equation}
  \{\pi_\Phi, \mu({\rm d}\Phi)\}
  \xrightarrow{\cP(\rho_0\otimes \cdot)}
  \{\omega_\Phi, \mu({\rm d}\Phi)\}
  \stackrel{\epsilon}{\approx} 
  \{\Phi(\rho_0), \mu({\rm d}\Phi)\}, 
\end{equation}
the first an ensemble on the program register $\cH_P$, the second the output of the processor, $\omega_\Phi = \cP(\rho_0\otimes\pi_\Phi)$, and the third the ideal ensemble from the implemented channels, if the processor was perfect. 

The Holevo information of the left-hand ensemble is upper bounded by $\log d_P^\infty$, and by data processing it is lower bounded by the middle one. We would like to apply a continuity bound for the von Neumann entropy to lower bound the Holevo information of the middle ensemble in turn in terms of the Holevo information of the right-hand ensemble. This is straightforward in finite dimension using the Fannes inequality, but a bit more subtle in infinite dimension, where however analogous bounds exist when additionally the states obey an energy bound~\cite{Winter16}. Assuming that $\cC \subset \cL(\alpha,\beta)$, this is indeed given for the states of the ideal ensemble: $\tr \Phi(\rho_0)H \leq \alpha E+\beta$. But we have a priori no such bound for the actual output states $\omega_\Phi$. 

\begin{lemma}
\label{lemming}
Consider two states $\rho,\sigma\in\cD(\cH)$, where $\cH$ carries a grounded Hamiltonian $H\geq 0$, and a number $E>0$. 
If $\frac12\|\rho-\sigma\|_1\leq\eta$ and $\tr\rho H \leq E$, then there exists a state $\sigma'$ with $\tr\sigma' H\leq \frac{E}{\eta}$ and
\begin{equation}
  \frac12\|\sigma-\sigma'\|_1 \leq 3\sqrt{\eta},
  \quad
  \frac12\|\rho-\sigma'\|_1 \leq 4\sqrt{\eta}.
\end{equation}
\end{lemma}

\begin{proof}
Take the subspace projector $P_{\eta}$ onto the energy subspace of all eigenvalues $\leq \frac{E}{\eta}$, and construct the compression map 
\[
  \cK(\xi) = P_{\eta}\xi P_{\eta} + \kappa_0 \tr\xi(\1-P_{\eta}),
\]
where $\kappa_0$ is an arbitrary state with support $P_{\eta}$, e.g. a ground state of $H$. Then, let $\sigma' \assign \cK(\sigma)$. This does it, as can be seen as follows:
$\tr\rho P_\eta \geq 1-\eta$, hence by the trace norm assumption, $\tr\sigma P_\eta \geq 1-2\eta$, and now we can apply the gentle operator lemma \ref{gentleoperator} and get $\frac12 \|\sigma-P_\eta\sigma P_\eta\|_1 \leq \sqrt{2\eta}$, hence by the triangle inequality $\frac12\|\sigma-\sigma'\|_1 \leq 2\eta+\sqrt{2\eta} \leq 3\sqrt{\eta}$. 

Using triangle inequality once more, we get the distance from $\rho$ bounded by $4\sqrt{\eta}$.
\end{proof}

\medskip
Using Lemma \ref{lemming}, we can process the ensemble further, letting $\omega_\Phi'=\cK(\omega_\Phi)$, using the compression map from the lemma: 
\begin{equation}
  \{\pi_\Phi, \mu({\rm d}\Phi)\}
  \xrightarrow{\cP(\rho_0\otimes\cdot)}
  \{\omega_\Phi, \mu({\rm d}\Phi)\}
  \xrightarrow{\,\cK\,}
  \{\omega_\Phi', \mu({\rm d}\Phi)\}
  \stackrel{4\sqrt{\epsilon}}{\approx} 
  \{\Phi(\rho_0), \mu({\rm d}\Phi)\},
\end{equation}
and now both $\Phi(\rho_0)$ and $\omega_\Phi'$ have their energy bounded by $\frac{\alpha E+\beta}{\epsilon} =: \widehat{E}$.
Assuming not only a grounded, but also finitary Hamiltonian $H_2$ with finite Gibbs entropy at all temperatures on the output space $\cH_2$, we then get the following chain of inequalities, lower bounding the program dimension:
\begin{equation}\begin{split}
  \label{useful-equation}
  \log d_P^\infty &\geq \chi\bigl(\{\pi_\Phi, \mu({\rm d}\Phi)\}\bigr) \\
        &\geq \chi\bigl(\{\omega_\Phi, \mu({\rm d}\Phi)\}\bigr) \\
        &\geq \chi\bigl(\{\omega_\Phi', \mu({\rm d}\Phi)\}\bigr) \\
        &\geq \chi\bigl(\{\Phi(\rho_0), \mu({\rm d}\Phi)\}\bigr)
                - 16\sqrt{\epsilon} S\left(\gamma\left(\frac{\widehat{E}}{4\sqrt{\epsilon}}\right)\right) - 2h(4\sqrt{\epsilon}) \\
        &\geq \chi\bigl(\{\Phi(\rho_0), \mu({\rm d}\Phi)\}\bigr)
                - 16\sqrt{\epsilon} S\left(\gamma\left(\frac{\alpha E+\beta}{4\sqrt{\epsilon}^{3}}\right)\right) - 2,
\end{split}\end{equation}
where the first three inequalities are by definition of the program ensemble and data processing (twice), and the fourth follows from Ref.~\cite[{Lemma~15}]{Winter16}. 
Here, $h(t)=-t\log t-(1-t)\log(1-t)$ is the binary entropy and we assume that $4\sqrt{\epsilon} \leq 1$.
(Alternatively, one could use the Meta-Lemma 16 from Ref.~\cite{Winter16}, which results in a similar bound with slightly better constants, but they are not that important for us, so we prefer the use of the simpler continuity bound.)

In the case of a class $\cC$ of unitary channels $\cU$, we first use Lemma \ref{lemma:E-Yang} to get a processor for the channels $\cU^{\otimes\ell}$, and then choose a fiducial state $\rho_0\in\cD(\cH^{\otimes\ell})$ such that $\tr\rho_0 H_j\leq E$ for all $j=1,\ldots,\ell$. The rest of the reasoning is then the same, except that to define the compression map in the application of Lemma \ref{lemming}, we have to consider total energy $\ell E$ at the input and $\ell(\alpha E+\beta)$ at the output.

\section{Programmable quantum processor for Gaussian \protect\\ channels}
\label{sec:gaussian}

We proved upper and lower bounds for a processor that implements all energy-limited channels with finite program register for an infinite-dimensional input state up to a certain energy $E>0$, i.e., $\tr(\rho H) \leq E$. In the following, the Hamiltonian is the photon-number operator $N\assign a^\ast a$. We now consider a special class of Gaussian channels: gauge-covariant Gaussian channels that are relevant in quantum optics, for instance. As explained in the previous sections, we also assume an energy constraint $\tr(\rho H) \leq E$ on the input. Thus, we already know from Section~\ref{section2} that there is a processor implementing an approximate version of all energy-limited Gaussian channels with finite-dimensional program register. These bounds also apply here. Since gauge-covariant Gaussian channels obey special properties compared to general channels, we aim to use those to develop better upper and lower bounds on the dimension of the program register.
We start with presenting the basic ingredients and fix the corresponding notation.

\subsection{Preliminaries and notation}
\label{subsec:Gaussian-stuff}
There are many review articles on Gaussian states and channels such as Refs.~\cite{Braunstein05, navarretebenlloch15, Weedbrook12} and textbooks such as Ref.~\cite{KokLovett}. Hence we review only briefly the most relevant notions and notations.

Density operators have an equivalent representation in terms of a Wigner function defined over the phase space, which is a real symplectic space $\mathbb{R}^{2M}$ equipped with the symplectic form.
The most fundamental object is the Weyl displacement operator defined as
\begin{equation}\label{displacementoperator}
\tilde D(\xi) \assign \exp (i X^T \Omega \xi),
\end{equation}
where $X=(x_1, p_1, \ldots, x_M, p_M)^T$ with $(x_j, p_j)$ the canonical quadrature operators for each mode $j \in \{1, \ldots, M\}$, 
\begin{equation}
    \Omega \assign \bigoplus_{j=1}^M \omega \text{ with } \omega=\begin{pmatrix} 0 & 1 \\ -1 & 0 \end{pmatrix}
\end{equation}
and $\xi \in \mathbb{R}^{2M}$. They establish a connection between operators and complex functions on phase space. \\
Then, an $M$-mode quantum state $\rho$ can be represented by its \emph{Wigner characteristic function} 
\begin{equation}
\tilde \chi_{\rho} (\xi) \assign \tr \rho \tilde D (\xi)
\end{equation}
for $\xi \in \mathbb{R}^{2M}$ and $\rho \in \cD(\cH^{\otimes M})$~\cite[Eq.~(12)]{Weedbrook12}.

\begin{definition}[Gaussian states]
\mylabel{Gaussianstates}
A Gaussian state is an $M$-mode quantum state with Gaussian characteristic function.
\end{definition}

An important example are \emph{coherent states} $\rho_G(\xi, \mathds{1}_2)$ of a single mode, $\xi=(\xi_1,\xi_2)$. They are pure and can be generated by displacing the vacuum state $\ket 0$, i.e., 
\begin{equation}
\ket \xi \assign D (\xi) \ket 0.
\end{equation}
The coherent states are the eigenstates of the annihilation operator,
\begin{equation}
a \ket \xi = \frac{\xi_1+i \xi_2}{\sqrt{2}} \ket \xi.
\end{equation}
They form an overcomplete basis, which means that any coherent state can be expanded in terms of all other coherent states because they are not orthogonal. Certain representations of states are based on a coherent-state expansion. Another class of Gaussian states are \emph{thermal states} $\rho_G ( 0 , (2N+1) \mathds{1}_2)$ with $N = \tr \rho a^\ast a$ being the mean photon number.

Let us now define Gaussian channels. 
\begin{definition}
A Gaussian channel is a quantum channel, i.e., a cptp map $\Phi : \cB (\cH^{\otimes n}) \to \cB ( \cH^{\otimes m})$ that maps every Gaussian state $\rho_G$ to a Gaussian state, i.e., $\Phi (\rho_G)$ is Gaussian as well. 
\end{definition}
Note that we can also input a non-Gaussian state into a Gaussian channel since Gaussianity is a property of the channel, not the state. The action of a Gaussian channel on a Gaussian state is described by the action on its first and second moments.
Gaussian channels acting on $M$ modes are characterized by a vector $\eta \in \mathbb{R}^{2M}$ and two real $2M \times 2M$ matrices $\cK$ and $\cN$. Those transform the displacement vector $d$ and the covariance matrix $\Gamma$ of the input state as follows
\begin{equation}
\label{Gaussianchanneltransformation}
d \to \cK d + \eta, \qquad \Gamma \to \cK \Gamma \cK^T + \cN.
\end{equation}
The matrices $\cK$ and $\cN$ obey the following relations: 
\begin{equation}
  \cN + i \Omega - i \cK \Omega \cK^T \geq 0,
\end{equation}
in order to map positive operators to positive operators and $\cN$ is symmetric. We can interpret $\cK$ as being responsible for the linear transformation of the canonical phase variables, while $\cN$ introduces quantum or classical noise. Recall that any channel can be conceived as a reduction of a unitary transformation acting on the system $\rho$ plus some environment $\rho_E$
\begin{equation}\mylabel{Gaussianunitary}
\Phi (\rho) = \tr_E U_G (\rho \otimes \rho_E) U_G^\ast.
\end{equation}
For a Gaussian channel $\Phi$, it is always possible to find a unitary dilation of this form with a  Gaussian unitary $U_S$ and a pure Gaussian environment state $\rho_E$. Special Gaussian unitaries are displacement operators [see Eq.~\eqref{displacementoperator}] and squeezing operators described by
\begin{equation}
    S(s)= \exp \left(\frac{s}{2} a^2 - \frac{s}{2} {a^\ast}^2 \right)
\end{equation}
with $s \in [0, \infty)$. It has zero displacement and its covariance matrix is
\begin{equation}
\begin{pmatrix} 
  e^{-s} & 0 \\ 0 & e^s 
\end{pmatrix}.
\end{equation}

As a special class of channels, we introduce \emph{gauge-covariant channels} which commute with the phase rotations. Since these channels are invariant under the rotation in phase space, they are also called \emph{phase-insensitive}.
\begin{definition}
A channel $\Phi$ that maps a single Bosonic mode to a single Bosonic mode is called gauge-covariant if it satisfies
\begin{equation}
\Phi ( e^{i\phi N_1} \rho e^{-i\phi N_1} ) = e^{i \phi N_2} \Phi (\rho) e^{-i \phi N_2}
\end{equation}
where $\phi$ is a real number and $N_{j} = a_j^\ast a_j$ is the photon-number operator in system $j=1,2$.
\end{definition}
Gauge-covariant channels are a concatenation of an attenuation and an amplification part. If one is interested in additivity questions of the Holevo information, for instance, then the parameter specifying the channel is considered to be real because the Holevo information is invariant under phase rotations which are passive transformations. However, viewed from the perspective of a processor, phases yield different outputs which we want to distinguish. Hence, we consider a complex parameter that specifies gauge-covariant channels.

The following proposition states this concatenation of an attenuator channel and an amplifier channel to obtain one-mode gauge-covariant quantum channels. Since we use complex parameters, we combine the phases to an additional phase-rotation part, which yields real parameters for the attenuation and amplification that can be found in Ref.~\cite{Holevobook}, for example. A quantum optical explanation including the complex parameter is given in Ref.~\cite[p.~31]{KokLovett}.

\begin{proposition}
\mylabel{propconcatenation} 
Any one-mode Bosonic gauge-covariant Gaussian channel $\Phi$ can be understood as a concatenation of a quantum-limited attenuator channel $\cT_\lambda$, a rotation channel $\cR_\varphi$ and a (diagonalizable) quantum-limited amplifier channel $\cA_\mu$, i.e., $\Phi = \cA_\mu \circ \cR_\varphi \circ \cT_\lambda$.
\end{proposition}

We briefly specify the three parts involved in the decomposition. First, there is an attenuator channel $\cT_\lambda$ with parameter $0 \leq \lambda \leq 1$, $\lambda \in \mathbb{R}$, which is called the \emph{attenuation factor}.
The two matrices describing the attenuator channel with the turn out to be $\cK = \sqrt{\lambda} \mathds{1}_2$ and $\cN = (1-\lambda) \mathds{1}_2$. The rotation in phase space is described by the unitary operator 
\begin{equation}
R(\varphi)=\exp(- i \varphi a^\ast a)
\end{equation}
which transforms the covariance matrix according to
\begin{equation}
\hat R =\begin{pmatrix} \phantom{-}
            \cos{\varphi} & \sin{\varphi} \\
           -\sin{\varphi} & \cos{\varphi}
        \end{pmatrix}.
\end{equation} 
We denote the corresponding rotation channel as $\cR_\varphi (\cdot) = e^{-i \varphi N} (\cdot) e^{i \varphi N} $.
The third part is an amplifier channel $\cA_\mu$ with parameter $\mu>1$. We consider quantum-limited amplifiers which are the least noisy deterministic amplifiers allowed by quantum mechanics. We obtain $\cK = \sqrt{\mu} \mathds{1}_2$ and $\cN = (\mu - 1) \mathds{1}_2$ for $\mu \in (1, \infty)$~\cite{navarretebenlloch15}.

\subsection{Gauge-covariant Gaussian channels}

In the first part, we considered an infinite-dimensional input state with a certain maximal energy $E$ and showed that there is a programmable quantum processor able to implement approximations of all $(\alpha, \beta)$-energy-limited unitary channels with finite-dimensional program register. 

In this section, we study the class of $(\alpha, \beta )$-energy-limited gauge-covariant Gaussian channels $\cG \cC \cG (\alpha, \beta)$. Recall that we denote the processor implementing these channels as $\epsilon\text{-EPQP}_{\cG \cC \cG (\alpha, \beta)}$. We give upper and lower bounds on the dimension of the program register of an approximate programmable quantum processor that implements all $(\alpha, \beta)$-energy-limited gauge-covariant Gaussian channels with an input and output state of a certain maximal energy.

\subsubsection{Upper bounds for gauge-covariant Gaussian channels}

To obtain upper bounds on the program dimension $d_P^\infty$, we establish an $\epsilon$-net on  $\cG \cC \cG (\alpha, \beta)$ to get a discrete approximation of the output. Afterwards, we use the PET to construct a processor that implements the channels of the $\epsilon$-net with program dimension equal to the cardinality of the $\epsilon$-net. With this construction, we obtain upper bounds on the program dimension of a processor implementing all $(\alpha, \beta)$-energy-limited gauge-covariant Gaussian channels.

\begin{theorem}[Upper bounds]\label{thm:upperboundsgaugecovariant}
Let $\epsilon > 0$ and $E > 0$. Then, there exists an infinite-dimensional $\epsilon\text{-EPQP}_{\cG\cC\cG(\alpha, \beta)}$ $\cP \in CPTP ( \cH_1 \otimes \cH_P, \cH_2)$ whose program register is upper bounded as follows:
\begin{equation}
d_P^\infty \leq \frac{C E^2 (2E+2) (\beta+1) }{\epsilon^6}
\end{equation}
for a constant $C$.
\end{theorem}

\begin{proof}
Since we consider gauge-covariant Gaussian channels, we use Proposition~\ref{propconcatenation} which states that those channels can be described as concatenation of an attenuator channel $\cT_\lambda$, a rotation channel $\cR_\varphi$ and a quantum-limited amplifier channel $\cA_{\mu}$ and thus, we construct one $\epsilon$-net on the set of attenuator channels, one on rotations and one on the amplifier channels. They are specified by one parameter each. 

Let us consider the attenuator channels first. Recall that the parameter $0 \leq \lambda< 1$ is the attenuation parameter. Thus, we construct an $\epsilon_\lambda$-net for $\lambda$ with $\{ \lambda_i \}_{i=1}^{|\cI_\lambda|} \subset [0, 1)$
such that for every $\lambda$ there is an index $i \in \cI_\lambda$ satisfying
\begin{equation}
|\lambda - \lambda_i | \leq \epsilon_\lambda.
\end{equation}
The range of $\lambda$ forms a compact interval. The cardinality of such a net is
\begin{equation}
|\cI_\lambda| \leq \bigg( \frac{1}{\epsilon_\lambda} + 1\bigg).
\end{equation}
Analogously, we construct an $\epsilon_\varphi$-net for the parameter $\varphi \in [0, 2\pi]$ such that for every $\varphi$, there exists an index $j \in \cI_\varphi$ with 
\begin{equation}
    |\varphi - \varphi_j| \leq \epsilon_\varphi
\end{equation}
with cardinality
\begin{equation}
|\cI_\varphi| \leq \bigg( \frac{2 \pi}{\epsilon_\varphi} +1 \bigg).
\end{equation}
We continue with the parameter describing the amplifier channel. Note that amplifier channels enlarge the energy. The larger the amplification factor, the higher the energy of the output. The $(\alpha, \beta)$-energy limitation of the considered channels yields a maximal amplification factor $\mu_{\text{max}}$. Let us specify this parameter.

To obtain a necessary condition for the parameter $\mu$, we consider the vacuum state $\rho_G (0,\mathds{1}_2)$ as input state with zero energy. The attenuator channel with $\cK=\sqrt{\lambda}\mathds{1}_2$ and $\cN=(1-\lambda)\mathds{1}_2$ and $\eta=0$ (see Subsection~\ref{subsec:Gaussian-stuff}) maps the vacuum state to the vacuum state. The amplifier channel with $\cK=\sqrt{\mu}\mathds{1}_2$ and $\cN=(\mu-1)\mathds{1}_2$, $\eta=0$ maps it to $\rho_G(0,(2\mu-1) \mathds{1}_2)$ with mean photon number
\begin{equation}
    \tr( \rho_G (0, (2\mu-1) \mathds{1}_2)) = \mu-1
\end{equation}
where we use $\tr(\rho_G a^\ast a)=\frac{1}{2} \tr(\Gamma) + \frac{1}{4}d^2 -\frac{1}{2}$ for a general $\rho_G(d, \Gamma)$~\cite[Eq.~(6.60)]{navarretebenlloch15}. This yields the necessary condition
\begin{equation}
    \mu \leq \beta +1.
\end{equation}
Hence, we choose 
\begin{equation}\label{eq:mumax}
    \mu_{max} =\beta +1.
\end{equation}

Due to the energy constraint and $\mu_{\text{max}}$, the values $\mu \in (1, \mu_{\text{max}}]$ form a compact set and we construct an $\epsilon_\mu$-net $\{ \mu_k\}_{k=1}^{|\cI_\mu|} \subset (1, \mu_{\text{max}}]$ such that for every $\mu$ there is an index $k \in \cI_{\mu}$ such that
\begin{equation}
|\mu - \mu_k | \leq \epsilon_\mu.
\end{equation}
The cardinality of this net reveals as
\begin{equation}
|\cI_\mu| \leq \bigg(  \frac{\mu_{\max}-1}{\epsilon_\mu}  + 1 \bigg).
 \end{equation}
The overall cardinality for the parameter appears as follows 
\begin{equation}
|\cI_{\Phi}| = |\cI_\lambda| |\cI_{\varphi}| |\cI_\mu| \leq \bigg( \frac{1}{\epsilon_\lambda} + 1\bigg) \bigg( \frac{2 \pi}{\epsilon_\varphi} +1 \bigg) \bigg(  \frac{\mu_{\text{max}}-1}{\epsilon_\mu}  + 1 \bigg) \leq \frac{16 \mu_{\text{max}}}{\epsilon_\lambda \epsilon_\varphi \epsilon_\mu} \leq  \frac{16 (\beta +1)}{\epsilon_\lambda \epsilon_\varphi \epsilon_\mu},
\end{equation}
where we use Eq.~\eqref{eq:mumax} in the last inequality.
Since we are interested in the cardinality of $\epsilon$-nets in $\cG \cC \cG (\alpha, \beta)$, we lift the parameter nets to nets on the set of channels. We use the $E$-diamond norm distance (see Definition~\ref{energydiamondnorm}). 

First, for the attenuator channel, we know from Ref.~\cite[Example 5]{Becker19} that
\begin{equation}
\|\cT_{\lambda} - \cT_{\lambda_i} \|_\diamond^E \leq 4 \sqrt{2} \sqrt{E \epsilon_\lambda}.
\end{equation}
Secondly, concerning the rotation channel, we use the result by Becker and Datta~\cite[Proposition 3.2]{Becker19} for the one-parameter unitary semigroup of rotations
\begin{equation}
\| \cR_\varphi - \cR_{\varphi_j}\|_\diamond^E \leq 4 \sqrt{E} \sqrt{|\varphi - \varphi_j|} = 4 \sqrt{E \epsilon_\varphi}.
\end{equation}
Thirdly, the norm of the distance of the amplifier channels can be bounded as~\cite[Example 5]{Becker19}
\begin{equation}
\|\cA_{\mu} - \cA_{\mu_k} \|_\diamond^E \leq 4 \sqrt{2} \sqrt{(2E+2) \epsilon_\mu}.
\end{equation}
For the $(\alpha, \beta)$-energy-limited gauge-covariant Gaussian channels we overall obtain
\begin{equation}
    \begin{split}
\| \Phi - \Phi_i \|_\diamond^E &\leq \|\cA_\mu \circ \cR_\varphi \circ \cT_\lambda - \cA_{\mu_k} \circ \cR_{\varphi_j} \circ \cT_{\lambda_i} \|_\diamond^E \\
& \leq  \| \cT_\lambda - \cT_{\lambda_i}  \|_\diamond^E + \|\cR_\varphi - \cR_{\varphi_j} \|_\diamond^E + \| \cA_\mu - \cA_{\mu_k} \|_\diamond^E  \\
&\leq 4 \sqrt{2} \sqrt{E} \sqrt{\epsilon_\lambda} + 4 \sqrt{E} \sqrt{\epsilon_\varphi} + 4 \sqrt{2} \sqrt{2E+2} \sqrt{\epsilon_\mu} \eqqcolon \epsilon.
    \end{split}
\end{equation}
We express $\epsilon_\lambda$, $\epsilon_\varphi$ and $\epsilon_\mu$ in terms of the $\epsilon$-parameter that specifies the accuracy of the processor:
\begin{equation}
\epsilon_\lambda = \frac{\epsilon^2}{C_\lambda E}, \qquad
    \epsilon_\varphi =\frac{\epsilon^2}{C_\varphi E}, \qquad \epsilon_\mu = \frac{\epsilon^2}{C_\mu ( 2E+2)}.
\end{equation}
Inserting these expressions into $|\cI_\Phi|$ we get
\begin{equation}
\begin{split}
    |\cI_\Phi|
    & \leq \frac{16 \  (\beta +1) \  C_\lambda \  C_\varphi \ C_\mu\  E (2E+2)}{\epsilon^6} = \frac{C E^2 (2E+2) \ (\beta +1)}{\epsilon^6}.
    \end{split}
\end{equation}
We use the PET to construct an $\epsilon\text{-EPQP}_{\cG\cC\cG(\alpha, \beta)}$ with program dimension 
\begin{equation}
  d_P^\infty=|\cI_\Phi| \leq \frac{C E^2 (2E+2) (\beta+1) }{\epsilon^6}, 
\end{equation}
concluding the proof.
\end{proof}

\subsubsection{Lower bounds for gauge-covariant Gaussian channels}
\label{subsec:lowerboundsgaugecovariant}

In Proposition~\ref{propconcatenation}, we obtained three different building blocks for the $\epsilon$-net for the upper bounds: attenuation, amplification, and phase rotation. It turns out that for lower bounds, the third part is particularly relevant because it yields $\epsilon$-divergence.

\paragraph{Phase rotation.}
To lower-bound the program dimension $d_P^\infty$ of an $\epsilon\text{-EPQP}_{\cG \cC \cG (\alpha, \beta)}$, we proceed in two steps. First, we apply Lemma~\ref{lemma:E-Yang} to the phase-rotation channels $\cR_\varphi(\cdot) = e^{-i\varphi N} (\cdot) e^{i\varphi N}$, which are $(1,0)$-limited. This results in a modified processor that implements the rotation $\ell$ times in parallel. 
Motivated by the fact that all information for the implementation of $\cR_\varphi$ is contained in the program state, which has to contain almost the same information as the $\ell$-tensor power phase rotation, we design an ensemble on the output space to obtain lower bounds on the program dimension by bounding the Holevo information in the second step. The resulting lower bounds are stated in the following theorem.

\begin{theorem}[Lower bounds phase rotations]
\label{thm:Gaussianlowerbounds}
Let $\epsilon > 0$ and $E > 0$. Then, for every infinite-dimensional $\epsilon\text{-EPQP}_{\cG\cC\cG(\alpha, \beta)}$ $\cP \in CPTP ( \cH \otimes \cH_P, \cH)$, with $\alpha\geq 1$ and $\beta\geq 0$, its program register can be lower bounded as follows: 
\begin{equation}
  d_P^\infty \geq \frac{1}{8192e}\frac{\delta^2 E}{(\sqrt{2}E+1)^\delta}\left(\frac{1}{\sqrt{2\epsilon}}\right)^{1-\delta} \geq C \delta^2 \left( \frac{E}{\sqrt{\epsilon}} \right)^{1-\delta},
\end{equation}
for any $0<\delta<1$, and the latter for $E\geq 1$ and an absolute constant $C$.
\end{theorem}

\begin{proof}
The gauge-covariant Gaussian channels contain the one-parameter group of unitaries $\cR_\varphi$ of phase rotations. 
Applying Lemma~\ref{lemma:E-Yang} to these unitaries, we obtain
\begin{equation}
    \label{eq:ell-phaserotations}
    \frac12 \left\|\widehat{\cP}(\cdot \otimes \psi_R) - \cR_\varphi^{\otimes \ell} \right\|_\diamond^{(E, \ldots, E)}  \leq 2 \ell \sqrt{2 \epsilon} .
\end{equation}

Next, we create a state of high total energy on which we act with $\cR_\varphi^{\otimes \ell}$ to generate an ensemble, where $\varphi$ is uniformly distributed on $[0, 2 \pi]$. 
The ideal output is the phase rotation $\cR_\varphi^{\otimes \ell}$ on $\ell$ modes. Since the photon number is the sum of those on the subsystems, the unitary is generated by the total photon number, i.e., $R_\varphi^{\otimes \ell} = e^{i \varphi (N_1 + \ldots + N_\ell)}$. Since this yields a phase multiplication $e^{i\varphi n}$ on each degenerate subspace of total photon number $n$, the $\ell$-fold tensor-product unitary is diagonal in the photon-number basis and we use only one state from each eigenspace. 

For each total number $n$ of photons, we define a unique way of distributing them across the $\ell$ modes in an as equilibrated way as possible. For instance, choose a partition of $n$ into non-negative integers, $n=n_1+\ldots+n_\ell$ and define $\ket{\text{``$n$''}}$ as the unit norm symmetrization of  $\ket{n_1}\cdots\ket{n_\ell}$,
\begin{equation}
    \ket{\text{``$n$''}} :\propto \sum_{\pi\in S_\ell} \ket{n_{\pi(1)}}\cdots\ket{n_{\pi(\ell)}}. 
\end{equation}
This state evidently has total photon number $n$, and the expected photon number in each mode is $\frac{n}{\ell}$.

As input state to the processor we choose
\begin{equation}
    \ket \nu =\sum_{n=0}^\infty c_n \ket{\text{``$n$''}},
\end{equation}
with amplitudes such that $\bra{\nu} N \ket{\nu} = \sum_n n|c_n|^2 \leq \ell E$.
The following calculations of the Holevo information take place on the subspace spanned by $\ket{\text{``$n$''}}$. Note that on that subspace, the total number operator $N = N_1+\ldots+N_\ell$ is isomorphic to a number operator, hence the time evolution of $\cR_\varphi^{\otimes\ell}$ leaves this ``virtual Fock space'' $\cH_V \subset \cH^{\otimes\ell}$ invariant. 

Since all information about the output of the processor is contained in the program state, 
\begin{equation}
    \log d_P^\infty \geq \chi\left(\left\{\widehat{\cP}(\proj{\nu}\otimes\psi_\varphi), \frac{{\rm d}\varphi}{2\pi}\right\}\right),
\end{equation}
following the approach explained at the end of Section \ref{sec:recycling}. We compare these ensemble states $\omega_\varphi=\widehat{\cP}(\proj{\nu}\otimes\psi_\varphi)$ with the ideal ones $\cR_\varphi^{\otimes\ell}(\proj{\nu})$, which are supported on $\cH_V$ with projector $P_V$. Thus, defining the compression map onto that subspace, 
\[
  \cK_V(\xi) = P_V\xi P_V + \kappa_0 \tr\xi(\1-P_V),
\]
and letting $\omega_\varphi'=\cK_V(\omega_\varphi)$, we have, by Eq.~\eqref{eq:ell-phaserotations} and the contractivity of the trace norm, that
\begin{equation}
  \frac12 \left\| \cR_\varphi^{\otimes\ell}(\proj{\nu}) - \omega_\varphi' \right\|_1 
  \leq \frac12 \left\| \cR_\varphi^{\otimes\ell}(\proj{\nu}) - \omega_\varphi \right\|_1 
  \leq 2\ell\sqrt{2\epsilon},
\end{equation}
and on the other hand by data processing for the Holevo information, 
\begin{equation}
    \log d_P^\infty \geq \chi\left(\left\{\omega_\varphi', \frac{{\rm d}\varphi}{2\pi}\right\}\right).
\end{equation}

As mentioned before, the Hamiltonian restricted to the virtual Fock space $\cH_V$ is isomorphic to a normal number operator $H_V=\sum_{n=0}^\infty n\proj{\text{``$n$''}}$, and the ideal ensemble states have energy $\tr \cR_\varphi^{\otimes\ell}(\proj{\nu}) H_V = \tr \proj{\nu} H_V \leq \ell E$. Now we invoke Lemma \ref{lemming} with $\eta = 2\ell\sqrt{2\epsilon}$, yielding the compression map $\cK$ onto the subspace with energy $\leq \frac{\ell E}{2\ell\sqrt{2\epsilon}} = \frac{E}{2\sqrt{2\epsilon}}$, which gives rise to states $\omega_\varphi'' = \cK(\omega_\varphi)$ with 
\begin{equation}
  \frac12 \left\| \cR_\varphi^{\otimes\ell}(\proj{\nu}) - \omega_\varphi'' \right\|_1 
  \leq 4 \sqrt{2\ell\sqrt{2\epsilon}},
\end{equation}
and so finally
\begin{equation}\begin{split}
  \log d_P^\infty
    &\geq \chi\left(\left\{\omega_\varphi'', \frac{{\rm d}\varphi}{2\pi}\right\}\right) \\
    &\geq \chi\left(\left\{\cR_\varphi^{\otimes\ell}(\proj{\nu}), \frac{{\rm d}\varphi}{2\pi}\right\}\right) 
          - 16 \sqrt{2} \sqrt{\ell\sqrt{2\epsilon}} g\left(\frac{E}{16\sqrt{\epsilon}\sqrt{\ell \sqrt{2\epsilon}}}\right) - 2 ,
\end{split}\end{equation}
using Eq.~\eqref{useful-equation} at the end of Section \ref{sec:recycling}, where $g(N)= (N+1)\log(N+1)-N\log N$ is the well-known formula for the von Neumann entropy of the thermal (Gibbs) state of mean photon number $N$. On the other hand, it is easily seen that 
\begin{equation}
  \chi\left(\left\{\cR_\varphi^{\otimes\ell}(\proj{\nu}), \frac{{\rm d}\varphi}{2\pi}\right\}\right) 
  = H(\{|c_n|^2\}), 
\end{equation}
which itself is maximized for the thermal distribution with mean photon number $\sum_n n|c_n|^2 = \ell E$, 
and the maximum is $g(\ell E)$. Using the elementary upper and lower bounds
\[
  \log N \leq g(N) \leq \log(N+1)+\log e, 
\]
and letting $\sqrt{\ell} = \frac{\delta}{16\sqrt{2\sqrt{2\epsilon}}}$, with $0<\delta<1$, we thus get 
\begin{equation}\begin{split}
  \log d_P^\infty
    &\geq \log(\ell E) - \delta\log\left(\frac{\sqrt{2} E}{\delta \sqrt{\epsilon}}+1\right) - \delta\log e - 2 \\
    &\geq (1-\delta)\log\frac{1}{\sqrt{2\epsilon}} + \log\frac{\delta^2 E}{512} - \delta\log(\sqrt{2}E+1)-\log(16e).
\end{split}\end{equation}
The final form of the bound is hence
\[
  d_P^\infty \geq \frac{1}{8192e}\frac{\delta^2 E}{(\sqrt{2}E+1)^\delta}\left(\frac{1}{\sqrt{2\epsilon}}\right)^{1-\delta} \geq C \delta^2 \left( \frac{E}{\sqrt{\epsilon}} \right)^{1-\delta},
\]
as claimed.
\end{proof}

Note that in this way, considering only the phase-rotation part of the decomposition of gauge-covariant Gaussian channels, we obtain lower bounds on the program dimension that diverge with $\epsilon$, confirming that the divergence of the upper bounds is not an artifact of the net construction.
Qualitatively, one can see that this is a method to obtain similar lower bounds on the program dimension to implement more general subgroups of unitaries.

\paragraph{Attenuation.} 
The lower bound we obtained in Theorem \ref{thm:Gaussianlowerbounds} relies only on the phase-rotation unitaries, since they allowed us to invoke the Replication Lemma \ref{lemma:E-Yang}. Here, we want to explore what kinds of lower bounds we can obtain from looking at attenuation only, i.e., on  $\epsilon\text{-EPQP}_\cT$, which denotes processors implementing the attenuator channels $\cT=\{\cR_\varphi\circ\cT_\lambda : \varphi\in\{0,\pi\},\, \lambda\in[0;1]\}$. Note that we allow a single phase rotation of angle $\pi$, which in itself cannot give an unbounded lower bound. We could give the subsequent argument without it, but including it makes the following discussion a little nicer.

\begin{theorem}
\label{thm:attenuator-lowerbounds}
Let $0<\epsilon<\frac{1}{1024}$ and $E \geq 2^e-1$. Then, for every infinite-dimensional $\epsilon\text{-EPQP}_\cT$ $\cP\in\text{CPTP}(\cH\otimes\cH_P,\cH)$, its program register is lower bounded as follows:
\begin{equation}
  d_P^\infty \geq 2^{-16} \frac{1}{\sqrt{\ln\log(E+1)}} (E+1)^{\frac12-16\sqrt{\epsilon}}. 
\end{equation}
\end{theorem}

\begin{proof}
As in the previous theorem, and explained at the end of Section \ref{sec:recycling}, to obtain the lower bound, we test the processor on a concrete input state $\rho_\chi$ which we choose to be a coherent state $\rho_\chi = \proj{\zeta}$ with the highest allowed energy (photon number) $E$, i.e., $\zeta = \sqrt{2E}$. With this fixed input state, the processor generates the output states $\{\Phi(\rho_\chi)\}$, $\Phi\in \cT$. 
These output states ideally are precisely the coherent states $\proj{\xi}$ with $-\zeta \leq \xi \leq \zeta$.

Rather than describing the ensemble of program states $\rho_\lambda$, which through the processor lead to unique output states $\rho_\xi = \cP(\proj{\zeta}\otimes\rho_{\lambda})$ (that approximate the coherent state $\proj{\xi}$), we give instead directly a distribution over the $\xi$. 
We choose the truncated Gaussian distribution with variance $\sigma^2$
\begin{equation}
  p_E(\xi) \assign 
  \begin{cases}
    \frac{1}{1-\eta}\frac{1}{\sqrt{2\pi}\sigma} e^{-\frac{\xi^2}{2\sigma^2}} & \text{ if } |\xi|\leq \sqrt{2E}, \\
    0 & \text{ otherwise}, 
  \end{cases}
\end{equation}
where 
\begin{equation}
  \label{eq:eta-bound}
  \eta = 1 - \int_{-\sqrt{2E}}^{+\sqrt{2E}} d\xi \frac{1}{\sqrt{2\pi}\sigma} e^{-\frac{\xi^2}{2\sigma^2}}
       =    \erfc(\sqrt{E}/\sigma)
       \leq e^{-E/\sigma^2},
\end{equation}
with the complementary error function $\erfc(x)$ and 
its well-known upper bound $\erfc(x) \leq e^{-x^2}$, see Ref.~\cite{Chiani:erf}.
Note furthermore that, denoting the density of the centered normal distribution with variance $\sigma^2$ by $p(\xi)=\frac{1}{\sqrt{2\pi}\sigma} e^{-\frac{\xi^2}{2\sigma^2}}$, we have 
$\frac12\|p-p_E\|_{L^1} = \eta$.

Following the method described at the end of Section \ref{sec:recycling}, using data processing, we start off from the inequality
\begin{equation}
\label{Infoprogramoutput}
  \log d_P^\infty \geq \chi (\{\rho_\xi, p_E(\xi) d\xi\})
           =    S\left(\int_{-\sqrt{2E}}^{\sqrt{2E}} d\xi\, p_E(\xi) \rho_\xi \right) - \int_{-\sqrt{2E}}^{\sqrt{2E}} d\xi\, p_E(\xi) S\left( \rho_\xi \right),
\end{equation}
where we recall the definition of the Holevo information.
The remaining calculation is about controlling the Holevo information on the right-hand side, which we do by first modifying the states from $\rho_\xi$ to the compressed state $\rho_\xi'=\cK(\rho_\xi)$, and finally to $\proj{\xi}$, incurring a certain error. According to Eq.~\eqref{useful-equation} we get 
from Eq.~\eqref{Infoprogramoutput}
\begin{equation}
\label{Infoidealoutput}
  \log d_P^\infty \geq S\left(\int_{-\sqrt{2E}}^{\sqrt{2E}} d\xi\, p_E(\xi) \proj{\xi} \right)
                - 16\sqrt{\epsilon} S\left(\gamma\left(\frac{E}{4\sqrt{\epsilon}^{3}}\right)\right) - 2,
\end{equation}
keeping in mind that our channels are $(1,0)$-energy-limited and that the attenuator output states $\proj{\xi}$ are pure.

It remains to calculate the entropy, which however is challenging. To lower bound it in turn, we modify the distribution from the truncated Gaussian $p_E$ to the full Gaussian $p$, incurring another certain Fannes error, but having the benefit of leaving us with a Gaussian state. We abbreviate the mixtures
\begin{align*}
  \omega   &\assign \int_{-\infty}^{\infty} d\xi\, \frac{1}{\sqrt{2\pi}\sigma} e^{-\frac{\xi^2}{2\sigma^2}} \proj{\xi}, \\
  \omega_E &\assign \int_{-\sqrt{2E}}^{\sqrt{2E}} d\xi\, p_E(\xi) \proj{\xi},
\end{align*}
to which we can apply~\cite[Lemma~15]{Winter16}, noting that both states have energy bounded by $\sigma^2/2$ and $E$, respectively. We choose $\sigma^2 = \frac{1}{t}E$ with $t \geq 1$, making the energy bound $E$, and $\eta \leq e^{-t}$, thus
\begin{equation}
|S(\omega) - S(\omega_E) | 
  \leq 2 \eta g \left(\frac{E}{\eta} \right) + h(\eta) 
  \leq 2 \eta \left( \log \left( \frac{E}{\eta}+1\right) + \log e \right) + 1.
\end{equation}
On the other hand, $\omega$ is a Gaussian state having a diagonal covariance matrix with eigenvalues $1$ and $1+2\sigma^2$. From this we can obtain its symplectic eigenvalue, which is $\sqrt{1+2\sigma^2}$, as one can see by considering a Gaussian squeezing unitary that transforms the state to a thermal Gaussian state. Hence,
\begin{equation}
  S(\omega) =    g\left( \frac{\sqrt{1+2\sigma^2}-1}{2} \right)
            \geq \log\left( \frac{\sqrt{1+2\sigma^2}+1}{2} \right),
\end{equation}
where here and in the previous display equation we used the bounds $\log(x+1) \leq g(x) \leq \log(x+1) + \log e$.

Putting it all together, using Eq.~\eqref{Infoidealoutput} and the above bounds, we obtain
\begin{equation}
\begin{split}
\log d_P^\infty
  &\geq g\left( \frac{\sqrt{1+2\sigma^2}-1}{2} \right) 
        - 2 \eta g \left(\frac{E}{\eta} \right) 
        - 16\sqrt{\epsilon} g\left(\frac{E}{4\sqrt{\epsilon}^{3}}\right) - 3 \\
  &\geq \log \left( \frac{\sqrt{1+2\sigma^2}+1}{2} \right) - 2 e^{-t} \log\left(E e^{t}+1\right) - 2 e^{-t} \log e \\
  &\phantom{=============}
        - 16 \sqrt{\epsilon} \log \left( \frac{E}{4\sqrt{\epsilon}^3}+1 \right) - 16 \sqrt{\epsilon} \log e -3 \\
  &\geq \frac12 \log\left( \frac{E+1}{2t} \right)
   - 2e^{-t}\log\bigl((E+1)e^t\bigr)
   - 16 \sqrt{\epsilon} \log \left( \frac{E+1}{4 \sqrt{\epsilon}^3} \right) \\
  &\phantom{==}
   - 3 - 2e^{-t}\log e - 16\sqrt{\epsilon}\log e \\
  &= \left(\frac12-16\sqrt{\epsilon}-2e^{-t} \right)\log(E+1) 
    - \frac12 \log t \\
  &\phantom{==}    
   - \frac12 - 3 - 2e^{-t}\log e - 2e^{-t}\log e^t + 16\sqrt{\epsilon}\log\left(\frac{4}{e}\sqrt{\epsilon}^3\right).
\end{split}
\end{equation}
Now we look at the terms in the last line, showing that their sum can be lower bounded by $-14$. Namely, note that the function $-x\log x$ is monotonically increasing on the interval $[0;1/e]$, and so $e^{-t}\log e^t \leq \frac{\log e}{e}$ as well as $-\sqrt{\epsilon}\log\sqrt{\epsilon} \leq \frac{5}{32}$, where we use that $\epsilon\leq\frac{1}{1024}$.

Thus, 
\begin{equation}
\log d_P^\infty
  \geq \left(\frac12-16\sqrt{\epsilon}\right)\log(E+1) 
    - \frac12 \log t - 14 - 2e^{-t}\log(E+1), 
\end{equation}
and letting $t=\ln\log(E+1) \geq 1$, recalling the assumption on $E$, concludes the proof.
\end{proof}

\paragraph{Amplification and attenuation.}
In the case of $\alpha > 1$, amplifier channels are also available. Thus, one could try to use attenuators as well as amplifiers to construct an ensemble. Heuristically, it makes sense to input a coherent state with the highest energy. Applying the attenuator channel yields coherent states with lower energy, which serve as input for the amplifier channel. The amplifier channel maps these coherent states to displaced thermal states. However, we did not find an appropriate ensemble where the mixture of ensemble states is still Gaussian (this is a heuristic to be able to calculate entropies in closed form) and the Holevo information is improved compared to the coherent states. 

Such an ensemble plausibly does not exist, because the amplifier channel introduces noise, which means that to obtain an advantage, the ensemble must use different amplification strengths, otherwise data processing shows directly that the Holevo information is only smaller. On the other hand, using a distribution over amplifications would likely result in an ensemble mixture that is a convex combination of different thermal states, and no nontrivial convex combination of them can lead to a thermal state again. 

Hence, to obtain lower bounds for gauge-covariant Gaussian channels, we consider the subset of attenuator channels. In other words, we obtain the same lower bounds for a processor that is merely able to implement attenuator channels and one that implements all Gaussian channels. 
Note that the lower bound on the processor dimension, such as Theorem \ref{thm:attenuator-lowerbounds}, in any case do not diverge as a function of $\epsilon$, as the net-based upper bounds do. Rather, for $\epsilon \to 0$ the lower bound converges to an expression that depends only on the energy, which is natural as it is based on an information bound for a single output of the channel that the processor implements, and we do not have access to the Replication Lemma \ref{lemma:E-Yang}.

\subsection{Gaussian unitary channels}

In analogy to existing programmable quantum processors that implement unitary channels, we aim for $(\alpha, \beta)$-energy-limited Gaussian unitary channels. Recall that we denote this processor as $\epsilon\text{-EPQP}_{\cG\cU(\alpha, \beta)}$. We study resource requirements in terms of the dimension of the program register.

\subsubsection{Upper bounds for Gaussian unitary channels}

We again aim to determine upper bounds on the program dimension $d_P^\infty$ of an $\epsilon\text{-EPQP}_{\cG\cU(\alpha, \beta)}$ in three steps. Firstly, we construct an $\epsilon$-net on the parameter set, secondly we relate it to a set of channels and thirdly, we construct a processor with program register equal to the cardinality of the net.
We denote the set of all $(\alpha, \beta)$-energy-limited Gaussian unitary channels as $\cG \cU(\alpha, \beta)$ with elements 
\begin{equation}
\cU_G (\cdot) = U_G \cdot U_G^\ast,
\end{equation}
where $S\in Sp_{2M}(\mathbb{R})$.

\begin{theorem}[Upper bound]\label{thm:upperboundsGaussianunitary}
Let $\epsilon > 0$ and $E > 0$. Then, for a system of $M$ Bosonic modes, there exists an infinite-dimensional $\epsilon\text{-EPQP}_{\cG\cU(\alpha, \beta)}$ $\cP \in CPTP (\cH\otimes \cH_P, \cH)$ whose program register is upper bounded as follows:
\begin{equation}
    |\cI_{S,d}| \leq  \bigg(\frac{2352 (M \alpha)^{3/2} (\sqrt{\alpha}+1)(E+1)}{\epsilon^2} \bigg)^{4M^2} \bigg( \frac{2\sqrt{2}(\sqrt{2\beta}+1) \sqrt{\alpha E + \beta + 1}}{ \epsilon } \bigg)^{2 M}
\end{equation}

for an absolute constant $c_S$.
\end{theorem}

\begin{proof}
A general $M$-mode Gaussian unitary can be decomposed into a generalized phase rotation, which is a passive operation that does not change the energy, followed by $M$ separate single-mode squeezing transformations, followed by another passive generalized rotation, and finally a displacement, cf. Refs.~\cite{Arvind-et-al,KariusBaktus}. Thus, we construct two nets: one for the first three operations based on the symplectic group and another one on the displacements. Note that these sets are not compact but due to the energy limitation, we can introduce a cutoff in each of the sets to obtain compact ones. Where we place the cutoff depends on $\alpha$ and $\beta$. 

For the first net, let us consider the following compact subset of $Sp_{2M}(\mathbb{R})$
\begin{equation}
\label{eq:symplectic-cutoff}
Sp_{2M}^{\sqrt{\alpha}+1} (\mathbb{R}) \assign \{ S \in Sp_{2M}(\mathbb{R}), \|S - \mathds{1}_{\mathbb{R}^{2M}}\|_\infty \leq \sqrt{\alpha} + 1\}.
\end{equation}
To see what this does, the cutoff yields all elements with maximal singular value $r\leq \sqrt{\alpha}$. Since the singular values of a symplectic matrix come in pairs $x$ and $1/x$, this means that the above subset contains all matrices whose singular values lie between $\frac{1}{\sqrt{\alpha}}$ and $\sqrt{\alpha}$. 
From the Bloch-Messiah decomposition, which shows that modulo passive Gaussian transformations, every Gaussian unitary described by a symplectic matrix is equivalent to a tensor product of single-mode squeezing operators~\cite{BlochMessiah}, this means that we only have to consider the $(\alpha,\beta)$-energy limitation of those $M$ single-mode squeezers.

It is indeed elementary to see that a squeezing unitary with singular values of the symplectic matrix $r\geq 1$ and $1/r$ is not $(\alpha,\beta)$-energy-limited for any $r^2>\alpha$ and $\beta\geq 0$. Indeed, consider a squeezed vacuum state $\sigma$ with $\Delta x^2=\gamma^2 \gg 1$, $\Delta p^2=\frac{1}{2\gamma^2}$, which has photon number $\tr\sigma H \approx \frac{\gamma^2-1}{2}$; after squeezing, we have an even more squeezed vacuum state with photon number $\tr\sigma'H \approx \frac{r^2\gamma^2-1}{2} \geq r^2 \tr\sigma H$.
On the other hand, it is not difficult to see that the squeezing unitary is $\left(r^2,\frac{r^2-1}{2}\right)$-energy-limited. Namely, it transforms $x$ to $rx$ and $p$ to $p/r$, and so the Hamiltonian $H=\frac12(x^2+p^2-1)$ is transformed to $H'=\frac12(r^2x^2+p^2/r^2-1) \leq \frac12 r^2(x^2+p^2-1) + \frac{r^2-1}{2} = r^2 H + \frac{r^2-1}{2}$.

Thus, the set $Sp_{2M}^{\sqrt{\alpha}+1} (\mathbb{R})$ from Eq.~\eqref{eq:symplectic-cutoff} contains all symplectic matrices giving rise to $(\alpha,\beta)$-energy-limited Gaussian unitaries.
An $\epsilon$-net $\cI_{S}$ on this set is constructed and its cardinality $|\cI_{S}|$ determined in Ref.~\cite[Lemma~S16]{Datta20} as
\begin{equation}
|\cI_{S}| \leq \left(\frac{3(\sqrt{\alpha}+1)}{\epsilon_S} \right)^{4M^2}.
\end{equation}

Concerning the displacement, we must have $|d|^2 = |d_1|^2 + |d_2|^2 + \ldots + |d_M|^2\leq 2\beta$ for an admissible displacement vector $d = (d_1,\ldots,d_M)$, where each $d_j=(d_{j1},d_{j2})$ is a pair of single-mode phase-space coordinates; otherwise the unitary channel $\cU_d$ is not $(\alpha,\beta)$-energy-limited. Indeed, consider a coherent state $\ket{\xi}=\ket{\xi_1}\cdots\ket{\xi_M}$ as input; the displacement unitary $D(d)$ transforms it into the coherent state $\ket{\xi+d} = \ket{\xi_1+d_1}\cdots\ket{\xi_M+d_M}$, which changes the photon number from $\frac12|\xi|^2$ to $\frac12|\xi+d|^2 = \frac12|\xi|^2+\frac12|d|^2+\xi^\top\cdot d$; choosing $\xi=0$ we see that necessarily $|d|^2\leq 2\beta$.

At the same time, the displacement unitary channel $\cU_d$ is $\left(1+t,\left(1+\frac{1}{t}\right)\frac12|d|^2\right)$-energy-limited for all $t>0$. Namely, for the $j$th mode, $D(d_j)$ transforms its annihilation operator $a_j$ to $a_j+\frac{1}{\sqrt{2}}(d_{j1}+id_{j2}) =: a_j + \alpha_j$; this has the effect of transforming the photon-number Hamiltonian $H_j=a_j^* a_j$ to 
\begin{equation}
  \label{eq:H-strich}
  H_j' = (a_j^* + \overline{\alpha}_j)(a_j + \alpha_j) 
     = a_j^* a_j + |\alpha_j|^2 + \overline{\alpha}_j a_j + \alpha_j a_j^*.
\end{equation}
Now consider that for $t>0$, we have 
\[
  0 \leq \left(\sqrt{t}a_j-\frac{1}{\sqrt{t}}\alpha_j\right)^* \left(\sqrt{t}a_j-\frac{1}{\sqrt{t}}\alpha_j\right) 
    = t a_j^* a_j + \frac{1}{t}|\alpha_j|^2
     - \overline{\alpha}_j a_j - \alpha_j a_j^*, 
\]
which we can insert into Eq.~\eqref{eq:H-strich} to get
\[
  H_j' 
   \leq (1+t) a_j^* a_j + \left(1+\frac1t\right)|\alpha_j|^2
   =    (1+t) H_j + \left(1+\frac1t\right)\frac12|d_j|^2.
\]
Summing over $j$ yields the claim.

Hence, we construct a net on the ball of radius $\sqrt{2 \beta}$ in $\mathbb{R}^{2M}$ with the Euclidean metric. Its cardinality is known to be~\cite[Lemma 5.8]{ABMB}
\begin{equation}
|\cI_{d}| \leq \left( 1 + \frac{\sqrt{2\beta}}{\epsilon_d} \right)^{2 M}.
\end{equation}

Bringing these two nets together by multiplying the unitaries, results in an overall net cardinality
\begin{equation}\begin{split}
|\cI_{S,d}| 
  \leq |\cI_S| |\cI_d| 
  &\leq \bigg(\frac{3(\sqrt{\alpha}+1)}{\epsilon_S} \bigg)^{4M^2} \bigg( 1 + \frac{\sqrt{2\beta}}{\epsilon_d} \bigg)^{2 M} \\
  &\leq \bigg(\frac{3(\sqrt{\alpha}+1)}{\epsilon_S} \bigg)^{4M^2} \bigg( \frac{\sqrt{2\beta}+1}{\epsilon_d} \bigg)^{2 M} .
\end{split}\end{equation}

Having established $\epsilon$-nets on the parameter level, we require an upper bound for $\frac{1}{2}\|\cU_G - \cU_{G_i}\|_\diamond^E$. So we transfer both nets on the sets of parameters to the corresponding channels. The symplectic matrices correspond to the first three parts of the decomposition, a rotation followed by a squeezing and again a rotation,

Since we assume an energy limitation on the set of Gaussian unitary channels, we can construct a compact subset of channels that contains these channels, as follows. In fact, we obtain it from a compact subset of the symplectic group and a compact subset of the displacement group.

For the former, we use Ref.~\cite[Eq.~(4)]{Datta20}, which states 
\begin{equation}
    \frac{1}{2}\|\cU_S - \cU_{S^\prime} \|_\diamond^E \leq \sqrt{(\sqrt{6}+\sqrt{10}+5\sqrt{2} M)(E+1)} g \bigl( \| S^{-1}S' \|_\infty \bigr) \sqrt{\|S^{-1}S' - \mathds{1}\|_2},
\end{equation}
where $g(x)\assign \sqrt{\frac{\pi}{x+1}} + \sqrt{2x}$.
Note $\|S^{-1}\| \leq \sqrt{\alpha}$, $\|S'\| \leq 1+\sqrt{\alpha}$, hence the argument $x$ of $g(x)$ is between $1$ and $\alpha+\sqrt{\alpha} \leq 2\alpha$, thus $g(x) \leq \sqrt{\pi/2} + \sqrt{2x} \leq \left(2 + \sqrt{\pi/2}\right)\sqrt{\alpha} < 3.26\sqrt{\alpha}$.
Furthermore, $\|S^{-1}S' - \mathds{1}\|_2 = \|S^{-1}(S' - S)\|_2 \leq \sqrt{2M} \|S^{-1}(S' - S)\|_\infty \leq \sqrt{2M} \|S^{-1}\|_\infty \|S' - S\|_\infty \leq \sqrt{2M}\sqrt{\alpha}\epsilon_S$.
Finally, $\sqrt{6}+\sqrt{10}+5\sqrt{2}M \leq (\sqrt{6}+\sqrt{10}+5\sqrt{2})M < 13M$.
Hence, in simplified form the result says that
\begin{equation}\begin{split}
\frac12 \| \cU_S - \cU_{S_i} \|_\diamond^E 
&\leq \sqrt{13M}\sqrt{E+1} \cdot 
      3.26\sqrt{\alpha} \cdot
      \sqrt[4]{2M}\sqrt[4]{\alpha}
      \sqrt{\epsilon_S} \\
&\leq 14 (M \alpha)^{3/4}\sqrt{E+1}\sqrt{\epsilon_S}.
\end{split}\end{equation}

For the latter, we consider Ref.~\cite[Eq.~(3)]{Datta20}, which states 
\begin{equation}
    \frac{1}{2}\|\cD_z - \cD_w \|_\diamond^E \leq \sin \left( \min \left\{ \| z - w \| f(E) , \frac{\pi}{2} \right\} \right)
    \leq \sqrt{2}\sqrt{E+1}\|z-w\|,
\end{equation}
where $f(E)\assign \frac{1}{\sqrt{2}} (\sqrt{E} + \sqrt{E+1}) \leq \sqrt{2}\sqrt{E+1}$ and as before we use only the simplified upper bound.

Note that with the action of $\cU_S$ on the input, the input energy of the displacement part changes, i.e.,~we work with the $(\alpha E + \beta)$-energy constrained diamond norm here. 
Thus we obtain
\begin{equation}
\begin{split}
\frac12 \| \cU_d - \cU_{d_i} \|_\diamond^{\alpha E + \beta} 
&\leq \|d - d_i \| f (\alpha E  + \beta)\\
&\leq \sqrt{2} \sqrt{\alpha E +\beta +1}  \epsilon_d.
\end{split}
\end{equation}

Overall, we obtain for the net of channels
\begin{equation}
\begin{split}
  \frac12\| \cU_G - \cU_{G_i}\|_\diamond^{E} &\leq \frac12\| \cU_d \circ \cU_S  - \cU_{d_i} \circ \cU_{S_i}\|_\diamond^E \\
&\leq  \frac12\| \cU_S - \cU_{S_i} \|_\diamond^E + \frac12\| \cU_d - \cU_{d_i}\|_\diamond ^{\alpha E + \beta} \\
& \leq 14 ( M \alpha)^{3/4}\sqrt{E+1}\sqrt{\epsilon_S} + \sqrt{2} \sqrt{\alpha E +\beta +1} \epsilon_d.
\end{split}
\end{equation}

We now choose $\epsilon_S$ and $\epsilon_d$ in terms of $\epsilon$, $E$, $\alpha$ and $\beta$, as follows:
\begin{equation}
\epsilon_S = \frac{\epsilon^2}{784 (M \alpha)^{3/ 2}(E+1)}, \qquad
    \epsilon_d = \frac{\epsilon }{ 2\sqrt{2} \sqrt{\alpha E + \beta +1}},
\end{equation}
which yields an $\epsilon$-net with
\begin{equation}
    |\cI_{S,d}| \leq  \bigg(\frac{2352 (M \alpha)^{3/2} (\sqrt{\alpha}+1)(E+1)}{\epsilon^2} \bigg)^{4M^2} \bigg( \frac{2\sqrt{2}(\sqrt{2\beta}+1) \sqrt{\alpha E + \beta + 1}}{ \epsilon } \bigg)^{2 M}.
\end{equation}
With the PET, we construct a programmable quantum processor with $d_P^\infty = |\cI_{S,d}|$ which concludes the proof.
\end{proof}

\subsubsection{Lower bounds for Gaussian unitary channels}

Concerning lower bounds on the dimension of the program register of an $\epsilon\text{-EPQP}_{\cG\cU(\alpha, \beta)}$ for an $M$-mode Bosonic system with Hilbert space $\cH=\cH_1\otimes\cdots\otimes\cH_M$, evidently Lemma \ref{lemma:E-Yang} is applicable to suitably energy-limited unitaries. The difficulty is mainly that of finding a good distribution over those unitaries and a fiducial state with which to calculate a lower bound along the lines of the end of Section \ref{sec:recycling}.

Since phase rotations are a subset of Gaussian unitaries, the bounds from Theorem~\ref{thm:Gaussianlowerbounds}  (Section~\ref{subsec:lowerboundsgaugecovariant}) directly apply here, showing that the program dimension diverges at least with the inverse square root of $\epsilon$. 
We can however immediately do a little better by considering $M$-fold phase rotations $\cR_{\und{\varphi}} = \cR_{\varphi_1}\otimes\cdots\otimes\cR_{\varphi_M}$
with $\und{\varphi}=(\varphi_1,\ldots,\varphi_M) \in [0;2\pi]^M$ in Lemma~\ref{lemma:E-Yang}. Note that these are a subset of the passive linear transformations, and as such conserve energy, hence are $(1,0)$-energy limited. 
We get that the modified processor $\widehat{\mathcal{P}}$ approximately implements the $\cR_{\und{\varphi}}^{\otimes\ell}$:
\[
  \frac12 \left\| \widehat{\mathcal{P}}(\cdot\otimes\pi_{\und{\varphi}})  -\cR_{\und{\varphi}}^{\otimes\ell} \right\|_\diamond^{(E,\ldots,E)} \leq 2\ell\sqrt{2\epsilon}.
\]
In the $\ell M$-mode system $\cH^{\otimes\ell}$ we 
address the modes $\cH_{jk}$ by double indices, where $j=1,\ldots,M$ are the original physical modes and $k=1,\ldots,\ell$ the repetitions. The repetitions of the $j$th mode have the Hilbert space $\cH_j^{\otimes\ell} = \cH_{j1}\otimes\cdots\otimes\cH_{j\ell} =: \cH_{j\bullet}$. As in the proof of Theorem~\ref{thm:Gaussianlowerbounds}, we choose a virtual Fock space $\cH_{V_j} \subset \cH_{j\bullet}$ spanned by symmetric number states $\ket{\text{``$n$''}}_j$, and let 
\[
  \ket{\nu_j} \assign \sum_{n=0}^\infty c_n \ket{\text{``$n$''}}_j,
\]
where $\{|c_n|^2\}$ is the probability distribution of the thermal state (of the virtual mode $\cH_{V_j}$) of mean energy (i.e., photon number) $\frac{\ell E}{M}$.
The rotation $\cR_{\varphi_j}^{\otimes\ell}$ acts on $\cH_{j\bullet}$ and leaves $\cH_{V_j}$ invariant, in fact it puts phases $e^{i\varphi_j n}$ in the above superposition. 
Now, $\ket{\nu} \assign \ket{\nu_1}\otimes\cdots\otimes\ket{\nu_M}$ is our fiducial state. It has the property that on each copy $\cH_{\bullet k} = \cH_{1k}\otimes\cdots\otimes\cH_{Mk}$ 
of the original $M$-mode system, its energy is $E$.
We can evaluate the Holevo information of the ideal ensemble of uniformly distributed states $\cR_{\und{\varphi}}^{\otimes\ell}(\proj{\nu})$ as we did in Theorem~\ref{thm:Gaussianlowerbounds}:
\begin{equation}
  \label{eq:M-mode-chi}
  \chi\left(\left\{ \cR_{\und{\varphi}}^{\otimes\ell}(\proj{\nu}), \frac{{\rm d}^M\und{\varphi}}{(2\pi)^M} \right\}\right)
  =    M g\left(\frac{\ell E}{M}\right)
  \geq M \log\frac{\ell E}{M}
  =    M\left(\log\ell + \log E - \log M\right). 
\end{equation}

By Lemma~\ref{lemma:E-Yang}, we now have for all $\und{\varphi}$
\[
  \frac12 \left\| \widehat{\mathcal{P}}(\proj{\nu}\otimes\pi_{\und{\varphi}})  -\cR_{\und{\varphi}}^{\otimes\ell}(\proj{\nu}) \right\|_1 \leq 2\ell\sqrt{2\epsilon}.
\]
Now we can massage the processor outputs $\omega_{\und{\varphi}} = \widehat{\mathcal{P}}(\proj{\nu}\otimes\pi_{\und{\varphi}})$ as in the proof of Theorem~\ref{thm:Gaussianlowerbounds}, first by the compression maps $\cK_j$ from $\cH_{j\bullet}$ to $\cH_{V_j}$, resulting in 
$\omega_{\und{\varphi}}' = (\cK_1\otimes\cdots\otimes\cK_M)\omega_{\und{\varphi}}$ obeying the same approximation as before
\[
  \frac12 \left\| \omega_{\und{\varphi}}' -\cR_{\und{\varphi}}^{\otimes\ell}(\proj{\nu}) \right\|_1 \leq 2\ell\sqrt{2\epsilon}.
\]
Second, by applying a compression map $\cK$ from $\cH_{V_1}\otimes\cdots\otimes\cH_{V_M}$ to an energy-constrained subspace, according to Lemma~\ref{lemming}, resulting in $\omega_{\und{\varphi}}'' = \cK(\omega_{\und{\varphi}}')$
such that
\[
  \frac12 \left\| \omega_{\und{\varphi}}'' -\cR_{\und{\varphi}}^{\otimes\ell}(\proj{\nu}) \right\|_1 \leq 4 \sqrt{ 2\ell\sqrt{2\epsilon}},
\]
while obeying an energy bound $\tr\omega_{\und{\varphi}}'' H \leq \frac{E}{2\sqrt{2\epsilon}}$. 

As explained at the end of Section~\ref{sec:recycling}, we now have
\begin{equation}\begin{split}
  \log d_P^\infty
    &\geq \chi\left(\left\{ \omega_{\und{\varphi}}'', \frac{{\rm d}^M\und{\varphi}}{(2\pi)^M} \right\}\right) \\
    &\geq \chi\left(\left\{ \cR_{\und{\varphi}}^{\otimes\ell}(\proj{\nu}), \frac{{\rm d}^M\und{\varphi}}{(2\pi)^M} \right\}\right)
          - 16 \sqrt{2} \sqrt{\ell}  \sqrt{\sqrt{2\epsilon}} M g\left(\frac{E}{16 \sqrt{\epsilon} \sqrt{\ell} \sqrt{\sqrt{2 \epsilon}}M} \right) - 2,
\end{split}\end{equation}
cf.~the proof of Theorem~\ref{thm:Gaussianlowerbounds}.
Inserting Eq.~\eqref{eq:M-mode-chi} for the Holevo information in the last line expresses everything in terms of the $g$ function, which we can upper and lower bound as before. We choose $\sqrt{\ell}=\frac{\delta}{16\sqrt{2} \sqrt{\sqrt{2\epsilon}}}$ with $0<\delta<1$, and obtain
\begin{equation}\begin{split}
  \log d_P^\infty
    &\geq M \log\frac{\ell E}{M} - \delta M  \log\left(\frac{\sqrt{2}E}{M\delta\sqrt{\epsilon}} + 1 \right) - \delta M\log e - 2 \\
    &\geq M(1-\delta)\log\frac{1}{\sqrt{2\epsilon}} + M\log\frac{\delta^2 E}{512M} - \delta M\log\left(\frac{\sqrt{2}E}{M}+1\right) - M \log e - 2.
\end{split}\end{equation}
This proves the following lower bound on the program dimension: 

\begin{theorem}[Lower bounds generalized phase rotations]
\label{thm:Gaussianlowerbounds-M-modes}
Let $\epsilon > 0$ and $E > 0$ and consider an $M$-mode Bosonic system with Hilbert space $\cH$. Then, for every infinite-dimensional $\epsilon\text{-EPQP}_{\cU(\alpha, \beta)}$ $\cP \in CPTP (\cH \otimes \cH_P, \cH)$, its program register can be lower bounded as follows: 
\begin{equation}
  d_P^\infty \geq \frac{1}{4(512 e)^M}\left(\frac{\delta^2 E/M}{(\sqrt{2}E/M+1)^\delta}\right)^M \left(\frac{1}{\sqrt{2\epsilon}}\right)^{(1-\delta)M}
      \geq (C \delta^2)^M \left(\frac{E/M}{\sqrt{\epsilon}}\right)^{(1-\delta)M},
\end{equation}
for any $0<\delta<1$, and the latter for $\frac{E}{M}\geq 1$ and an absolute constant $C$.
\hfill $\square$
\end{theorem}

\medskip
In the above analysis, as in that of Theorem~\ref{thm:Gaussianlowerbounds}, it is helpful that the ensemble is obtained from an orbit of a compact subgroup of the Gaussian unitaries, for which we can impose the Haar measure as probability distribution. It might be possible to extend the analysis of Theorem~\ref{thm:Gaussianlowerbounds-M-modes} to the largest compact subgroup, which is the set of all passive linear transformations, parametrized precisely by the orthogonal symplectic matrices $Sp_{2M}(\mathbb{R})\cap SO_{2M}(\mathbb{R})$, which is well known to be isomorphic to the group $\cU(\mathbb{C}^M)$ of $M\times M$-unitary matrices (cf. Ref.~\cite{Arvind-et-al}). This would give access to a much larger set of ensembles, but in view of the above result (featuring the entropy-maximizing Gibbs states), it would be surprising if it yielded a significant improvement.

\section{The future}
\label{sec:conclusion}

We introduce approximate programmable quantum processors in infinite dimension and establish a relation to finite-dimensional ones. With two different constructions, the first reducing an infinite-dimensional processor to a finite-dimensional one, and the second vice versa, we obtain upper and lower bounds on the program dimension based on existing bounds. It remains an open question to determine how tight these bounds are, and may best be addressed in a specific setting. 

We do this for Bosonic systems and restrict to Gaussian channels. In the unitary case, we obtain lower bounds that diverge with the reciprocal of the approximation parameter, with a polynomial degree of at least half the number of Bosonic modes in the system. We also have upper bounds based on nets on compact sections of the Gaussian unitary group, also polynomial in the reciprocal of the approximation parameter, but the order is quadratic in the number of modes. We leave the determination of the exact optimal degree as an open question. To obtain these lower bounds, the method of Yang \emph{et al.}~\cite{Renner20} by which the program state can be recycled to implement the same unitary several times, is adapted to infinite dimension and the energy-constrained diamond norm (necessitating a generalization of the latter with multiple constraints).

However, this trick cannot be applied in the case of the  well-studied and physically relevant class of single-mode gauge-covariant Gaussian channels, because those are genuinely noisy channels. We still provide lower bounds, which are comparatively weak, though. Concretely, while the corresponding upper bounds that we develop (based on nets) diverge with the accuracy of implementation, the lower bounds tend to a constant function of the energy constraint. It remains an open problem to close this gap, and indeed to determine whether an infinite program register is actually necessary to implement all gauge-covariant channels (modulo phase rotations), or even only the attenuator channels. The latter is a particularly interesting subclass, as attenuators map pure coherent states to other pure coherent states, suggesting that perhaps a variant of Lemma \ref{lemma:E-Yang} could be shown to hold for repeated applications of the channel to coherent states. 

Furthermore, one could think of more degrees of generalizations towards a programma-\\ble quantum processor that implements all Gaussian channels. This could potentially be done by considering their Gaussian unitary dilations, for which we already have upper bounds. A necessary requirement to carry out this program would be a version of the Stinespring theorem that gives an energy-limited Gaussian dilation for every energy-limited Gaussian channel.

\vfill\pagebreak
\myacknowledgements

M.G. thanks Andreas Bluhm, Stefan Huber, Cambyse Rouz\'{e} and Simone Warzel for discussions. M.G. is funded by the Deutsche Forschungsgemeinschaft (DFG, German Research Foundation) under Germany's Excellence Strategy -- EXC-2111 -- 390814868.
A.W. acknowledges financial support by the Spanish MINECO (Projects FIS2016-86681-P and PID2019-107609GB-I00) with the support of FEDER funds, and the Generalitat de Catalunya (Project CIRIT 2017-SGR-1127).

\bibliographystyle{unsrturl}

\bibliography{mybib}

\begin{thebibliography}{10}

\bibitem{Nielsen97}
M.~A. {Nielsen} and I.~L. {Chuang}.
\newblock {Programmable Quantum Gate Arrays}.
\newblock {\em Physical Review Letters}, 79(2):321--324, 1997.
\newblock \href {https://doi.org/10.1103/PhysRevLett.79.321}
  {\path{doi:10.1103/PhysRevLett.79.321}}.

\bibitem{Renner20}
Y.~Yang, R.~Renner, and G.~Chiribella.
\newblock {Optimal Universal Programming of Unitary Gates}.
\newblock {\em Physical Review Letters}, 125:210501, 2020.
\newblock arXiv:2007.10363v3 [quant-ph].
\newblock \href {https://doi.org/10.1103/PhysRevLett.125.210501}
  {\path{doi:10.1103/PhysRevLett.125.210501}}.

\bibitem{Weedbrook12}
C.~Weedbrook, S.~Pirandola, R.~Garc\'ia-Patr\'on, N.~J. Cerf, T.~C. Ralph,
  J.~H. Shapiro, and S.~Lloyd.
\newblock {Gaussian quantum information}.
\newblock {\em Reviews of Modern Physics}, 84(2):621–669, 2012.
\newblock \href {https://doi.org/10.1103/revmodphys.84.621}
  {\path{doi:10.1103/revmodphys.84.621}}.

\bibitem{GGCH}
V.~Giovannetti, R.~Garcia-Patron, N.~J. Cerf, and A.~S. Holevo.
\newblock Ultimate communication capacity of quantum optical channels by
  solving the {G}aussian minimum-entropy conjecture.
\newblock {\em Nature Photonics}, 8:796--800, 2014.
\newblock \href {https://doi.org/10.1038/nphoton.2014.216}
  {\path{doi:10.1038/nphoton.2014.216}}.

\bibitem{Grosshans-et-al}
F.~Grosshans, G.~{Van Assche}, J.~Wenger, R.~Brouri, N.~J. Cerf, and
  P.~Grangier.
\newblock Quantum key distribution using {G}aussian-modulated coherent states.
\newblock {\em Nature}, 421:238–241, 2003.
\newblock \href {https://doi.org/10.1038/nature01289}
  {\path{doi:10.1038/nature01289}}.

\bibitem{Leverrier}
A.~Leverrier.
\newblock {Composable Security Proof for Continuous-Variable Quantum Key
  Distribution with Coherent States}.
\newblock {\em Physical Review Letters}, 114:070501, 2015.
\newblock \href {https://doi.org/10.1103/PhysRevLett.114.070501}
  {\path{doi:10.1103/PhysRevLett.114.070501}}.

\bibitem{KLM}
E.~Knill, R.~Laflamme, and G.~J. Milburn.
\newblock A scheme for efficient quantum computation with linear optics.
\newblock {\em Nature}, 409:46--52, 2001.
\newblock \href {https://doi.org/10.1038/35051009}
  {\path{doi:10.1038/35051009}}.

\bibitem{BosonSampl}
S.~Aaronson and A.~Arkhipov.
\newblock {The Computational Complexity of Linear Optics}.
\newblock In {\em Proc. STOC'11, June 6–8, 2011, San Jose, California, USA},
  pages 333--342, 2011.
\newblock \href {https://doi.org/10.1145/1993636.1993682}
  {\path{doi:10.1145/1993636.1993682}}.

\bibitem{GKP}
K.~Noh, C.~Chamberland, and F.~G. S.~L. Brand{\~a}o.
\newblock Low overhead fault-tolerant quantum error correction with the
  surface-{GKP} code.
\newblock {\em arXiv:2103.06994 [quant-ph]}, 2015.
\newblock URL: \url{https://arxiv.org/pdf/2103.06994.pdf}.

\bibitem{Shirokov18}
M.~E. Shirokov.
\newblock {{On the Energy-Constrained Diamond Norm and Its Application in
  Quantum Information Theory}}.
\newblock {\em Problems of Information Transmission}, 54(1):20–33, 2018.
\newblock \href {https://doi.org/10.1134/s0032946018010027}
  {\path{doi:10.1134/s0032946018010027}}.

\bibitem{Winter17}
A.~Winter.
\newblock {Energy-constrained diamond norm with applications to the uniform
  continuity of continuous variable channel capacities}.
\newblock {\em arXiv:1712.10267 [quant-ph]}, 2017.
\newblock URL: \url{https://arxiv.org/pdf/1712.10267.pdf}.

\bibitem{Lupo17}
S.~Pirandola and C.~Lupo.
\newblock {Ultimate Precision of Adaptive Noise Estimation}.
\newblock {\em Physical Review Letters}, 118:100502, 2017.
\newblock \href {https://doi.org/10.1103/physrevlett.118.100502}
  {\path{doi:10.1103/physrevlett.118.100502}}.

\bibitem{Pirandola17}
S.~Pirandola, R.~Laurenza, C.~Ottaviani, and L.~Banchi.
\newblock {Fundamental limits of repeaterless quantum communications}.
\newblock {\em Nature Communications}, 8(1):15043, 2017.
\newblock \href {https://doi.org/10.1038/ncomms15043}
  {\path{doi:10.1038/ncomms15043}}.

\bibitem{Bennett96}
C.~H. {Bennett}, D.~P. {Divincenzo}, J.~A. {Smolin}, and W.~K. {Wootters}.
\newblock {Mixed-state entanglement and quantum error correction}.
\newblock {\em Physical Review A}, 54(5):3824--3851, 1996.
\newblock \href {https://doi.org/10.1103/PhysRevA.54.3824}
  {\path{doi:10.1103/PhysRevA.54.3824}}.

\bibitem{Takeoka16}
M.~Takeoka and M.~M. Wilde.
\newblock {Optimal estimation and discrimination of excess noise in thermal and
  amplifier channels}.
\newblock {\em arXiv:1611.09165 [quant-ph]}, 2016.
\newblock URL: \url{https://arxiv.org/pdf/1611.09165.pdf}.

\bibitem{Pirandola-et-al-2018}
S.~Pirandola, B.~R. Bardhan, T.~Gehring, C.~Weedbrook, and S.~Lloyd.
\newblock {Advances in Photonic Quantum Sensing}.
\newblock {\em Nature Photonics}, 12:724--733, 2018.
\newblock \href {https://doi.org/10.1038/s41566-018-0301-6}
  {\path{doi:10.1038/s41566-018-0301-6}}.

\bibitem{Berta18}
M.~M. Wilde, M.~Berta, C.~Hirche, and E.~Kaur.
\newblock {Amortized Channel Divergence for Asymptotic Quantum Channel
  Discrimination}.
\newblock {\em Letters in Mathematical Physics}, 110:2277--2336, 2020.
\newblock \href {https://doi.org/10.1007/s11005-020-01297-7}
  {\path{doi:10.1007/s11005-020-01297-7}}.

\bibitem{vonNeumann:book}
J.~{von Neumann (transl. R. T. Beyer)}.
\newblock {\em Mathematical Foundations of Quantum Mechanics}.
\newblock Princeton University Press, Princeton, 2018.
\newblock {Original German edition, \emph{Mathematische Grundlagen der
  Quantenmechanik}, Springer Verlag, 1932}.
\newblock \href {https://doi.org/10.1007/978-3-642-61409-5}
  {\path{doi:10.1007/978-3-642-61409-5}}.

\bibitem{Hall:QM-x-Mathematicians}
B.~C. Hall.
\newblock {\em Quantum Theory for Mathematicians}, volume 267 of {\em Graduate
  Texts in Mathematics}.
\newblock Springer Verlag, Berlin Heidelberg New York, 2013.
\newblock \href {https://doi.org/10.1007/978-1-4614-7116-5}
  {\path{doi:10.1007/978-1-4614-7116-5}}.

\bibitem{Kubicki19}
A.~M. {Kubicki}, C.~{Palazuelos}, and D.~{P{\'e}rez-Garc{\'\i}a}.
\newblock {Resource Quantification for the No-Programing Theorem}.
\newblock {\em Physical Review Letters}, 122(8):080505, 2019.
\newblock \href {https://doi.org/10.1103/PhysRevLett.122.080505}
  {\path{doi:10.1103/PhysRevLett.122.080505}}.

\bibitem{Majenz}
C.~Majenz.
\newblock {\em Entropy in Quantum Information Theory -- Communication and
  Cryptography}.
\newblock PhD thesis, PhD School of The Faculty of Science, University of
  Copenhagen, 2018.
\newblock URL: \url{https://arxiv.org/pdf/1810.10436.pdf}.

\bibitem{Winter99}
A.~Winter.
\newblock {Coding theorem and strong converse for quantum channels}.
\newblock {\em IEEE Transactions on Information Theory}, 45(7):2481–2485,
  1999.
\newblock \href {https://doi.org/10.1109/18.796385}
  {\path{doi:10.1109/18.796385}}.

\bibitem{Ogawa02}
T.~{Ogawa} and H.~{Nagaoka}.
\newblock {A New Proof of the Channel Coding Theorem via Hypothesis Testing in
  Quantum Information Theory}.
\newblock {\em {Proceedings IEEE International Symposium on Information
  Theory}}, 2002.
\newblock \href {https://doi.org/10.1109/ISIT.2002.1023345}
  {\path{doi:10.1109/ISIT.2002.1023345}}.

\bibitem{Wilde13}
M.~M. Wilde.
\newblock {\em Quantum Information Theory}.
\newblock Cambridge University Press, Cambridge, 1st edition, 2013.
\newblock \href {https://doi.org/10.1017/CBO9781139525343}
  {\path{doi:10.1017/CBO9781139525343}}.

\bibitem{Port-based08}
S.~{Ishizaka} and T.~{Hiroshima}.
\newblock {Asymptotic Teleportation Scheme as a Universal Programmable Quantum
  Processor}.
\newblock {\em Physical Review Letters}, 101(24):240501, 2008.
\newblock \href {https://doi.org/10.1103/PhysRevLett.101.240501}
  {\path{doi:10.1103/PhysRevLett.101.240501}}.

\bibitem{Beigi11}
S.~Beigi and R.~K{\"o}nig.
\newblock {Simplified instantaneous non-local quantum computation with
  applications to position-based cryptography}.
\newblock {\em New Journal of Physics}, 13(9):093036, 2011.
\newblock \href {https://doi.org/10.1088/1367-2630/13/9/093036}
  {\path{doi:10.1088/1367-2630/13/9/093036}}.

\bibitem{Christiandl18}
M.~{Christandl}, F.~{Leditzky}, C.~{Majenz}, G.~{Smith}, F.~{Speelman}, and
  M.~{Walter}.
\newblock {Asymptotic performance of port-based teleportation}.
\newblock {\em Communications in Mathematical Physics}, 381(1):379--451, 2020.
\newblock \href {https://doi.org/10.1007/s00220-020-03884-0}
  {\path{doi:10.1007/s00220-020-03884-0}}.

\bibitem{Pirandola19}
S.~{Pirandola}, R.~{Laurenza}, C.~{Lupo}, and J.~L. {Pereira}.
\newblock {Fundamental limits to quantum channel discrimination}.
\newblock {\em npj Quantum Information}, 5:50, 2019.
\newblock \href {https://doi.org/10.1038/s41534-019-0162-y}
  {\path{doi:10.1038/s41534-019-0162-y}}.

\bibitem{Perez-Garcia06}
D.~Pérez-García.
\newblock {Optimality of programmable quantum measurements}.
\newblock {\em Physical Review A}, 73(5):052315, 2006.
\newblock \href {https://doi.org/10.1103/physreva.73.052315}
  {\path{doi:10.1103/physreva.73.052315}}.

\bibitem{Shirokov16}
M.~E. Shirokov.
\newblock {Uniform continuity bounds for information characteristics of quantum
  channels depending on input dimension and on input energy}.
\newblock {\em Journal of Physics A: Mathematical and Theoretical},
  52(1):014001, 2018.
\newblock \href {https://doi.org/10.1088/1751-8121/aaebac}
  {\path{doi:10.1088/1751-8121/aaebac}}.

\bibitem{KSW:dilation2}
D.~Kretschmann, D.~Schlingemann, and R.~F. Werner.
\newblock A continuity theorem for stinespring’s dilation.
\newblock {\em Journal of Functional Analysis}, 255(8):1889–1904, 2008.
\newblock \href {https://doi.org/10.1016/j.jfa.2008.07.023}
  {\path{doi:10.1016/j.jfa.2008.07.023}}.

\bibitem{Winter16}
A.~Winter.
\newblock {Tight Uniform Continuity Bounds for Quantum Entropies: Conditional
  Entropy, Relative Entropy Distance and Energy Constraints}.
\newblock {\em Communications in Mathematical Physics}, 347(1):291--313, 2016.
\newblock \href {https://doi.org/10.1007/s00220-016-2609-8}
  {\path{doi:10.1007/s00220-016-2609-8}}.

\bibitem{Braunstein05}
S.~L. Braunstein and P.~{van Loock}.
\newblock {Quantum information with continuous variables}.
\newblock {\em Reviews of Modern Physics}, 77(2):513–577, 2005.
\newblock \href {https://doi.org/10.1103/revmodphys.77.513}
  {\path{doi:10.1103/revmodphys.77.513}}.

\bibitem{navarretebenlloch15}
C.~Navarrete-Benlloch.
\newblock {\em {An introduction to the formalism of quantum information with
  continuous variables }}.
\newblock IOP release 2. Morgan \& Claypool Publishers, San Rafael, California,
  2015.
\newblock \href {https://doi.org/10.1088/978-1-6817-4405-6}
  {\path{doi:10.1088/978-1-6817-4405-6}}.

\bibitem{KokLovett}
P.~Kok and B.~W. Lovett.
\newblock {\em {Introduction to Optical Quantum Information Processing}}.
\newblock Cambridge University Press, Cambridge, 2010.
\newblock \href {https://doi.org/10.1017/CBO9781139193658}
  {\path{doi:10.1017/CBO9781139193658}}.

\bibitem{Holevobook}
A.~S. Holevo.
\newblock {\em Quantum Systems, Channels, Information: A Mathematical
  Introduction}.
\newblock Texts and Monographs in Theoretical Physics. De Gruyter, Berlin,
  Boston, 2nd edition, 2019.
\newblock \href {https://doi.org/10.1515/9783110273403}
  {\path{doi:10.1515/9783110273403}}.

\bibitem{Becker19}
S.~Becker and N.~Datta.
\newblock {Convergence Rates for Quantum Evolution and Entropic Continuity
  Bounds in Infinite Dimensions}.
\newblock {\em Communications in Mathematical Physics}, 374(2):823–871, 2019.
\newblock \href {https://doi.org/10.1007/s00220-019-03594-2}
  {\path{doi:10.1007/s00220-019-03594-2}}.

\bibitem{Chiani:erf}
M.~Chiani, D.~Dardari, and M.~K. Simon.
\newblock {New Exponential Bounds and Approximations for the Computation of
  Error Probability in Fading Channels}.
\newblock {\em IEEE Transactions on Wireless Communications}, 2(4):840--845,
  2003.
\newblock \href {https://doi.org/10.1109/TWC.2003.814350}
  {\path{doi:10.1109/TWC.2003.814350}}.

\bibitem{Arvind-et-al}
Arvind, B.~Dutta, N.~Mukunda, and R.~Simon.
\newblock {The real symplectic groups in quantum mechanics and optics}.
\newblock {\em Pramana -- Journal of Physics}, 45(6):471--497, 1995.
\newblock \href {https://doi.org/10.1007/BF02848172}
  {\path{doi:10.1007/BF02848172}}.

\bibitem{KariusBaktus}
G.~Cariolaro and G.~Pierobon.
\newblock {{Bloch-Messiah reduction of Gaussian unitaries by Takagi
  factorization}}.
\newblock {\em Physical Review A}, 94:062109, 2016.
\newblock \href {https://doi.org/10.1103/PhysRevA.94.062109}
  {\path{doi:10.1103/PhysRevA.94.062109}}.

\bibitem{BlochMessiah}
C.~Bloch and A.~Messiah.
\newblock {The canonical form of an antisymmetric tensor and its application to
  the theory of superconductivity}.
\newblock {\em Nuclear Physics}, 39:95--106, 1962.
\newblock \href {https://doi.org/10.1016/0029-5582(62)90377-2}
  {\path{doi:10.1016/0029-5582(62)90377-2}}.

\bibitem{Datta20}
S.~Becker, N.~Datta, L.~Lami, and C.~Rouzé.
\newblock {Energy-Constrained Discrimination of Unitaries, Quantum Speed
  Limits, and a Gaussian Solovay-Kitaev Theorem}.
\newblock {\em Physical Review Letters}, 126(19), 2021.
\newblock \href {https://doi.org/10.1103/physrevlett.126.190504}
  {\path{doi:10.1103/physrevlett.126.190504}}.

\bibitem{ABMB}
G.~Aubrun and S.~J. Szarek.
\newblock {\em Alice and Bob meet Banach: The Interface of Asymptotic Geometric
  Analysis and Quantum Information Theory}, volume 223 of {\em Mathematical
  surveys and monographs}.
\newblock American Mathematical Society, Providence, Rhode Island, 2017.
\newblock \href {https://doi.org/10.1090/surv/223}
  {\path{doi:10.1090/surv/223}}.

\end{thebibliography}

\end{document}